%% file: A_main.tex
\documentclass[conference]{IEEEtran}
\IEEEoverridecommandlockouts
\usepackage{graphicx}
\usepackage{textcomp}
\usepackage{booktabs}
\usepackage{amsmath}
\usepackage{amsfonts}
\usepackage{amsthm}
\usepackage{thmtools, thm-restate}
\usepackage{algorithm}
\usepackage{algpseudocode}
\usepackage{csquotes}
\usepackage{enumitem}
\usepackage{multirow}
\usepackage{multicol}
\usepackage{tabularx}
\usepackage{array}
\usepackage{xcolor}
\usepackage{subcaption}
\usepackage{flushend}
\usepackage{comment}
\usepackage{mathrsfs}
\usepackage{gensymb}
\usepackage{bbold}
\usepackage{nicefrac}
\usepackage{nicematrix}
\usepackage[normalem]{ulem}

\newif\ifediting
\editingtrue  
\input{macros/comments.tex} 
\input{macros}
\newtheorem{definition}{Definition}

\newtheorem{corollary}{Corollary}
\newtheorem{lemma}{Lemma}

\newenvironment{spmatrix}[1]
 {\def\mysubscript{#1}\mathop\bgroup\begin{pmatrix}}
 {\end{pmatrix}\egroup_{\textstyle\mathstrut\mysubscript}}

\def \calA{{\mathcal{A}}}
\def \calX{{\mathcal{X}}}
\def \calY{{\mathcal{Y}}}

\def \setOne{\mathbb{1}}
\def \eeps{e^{\varepsilon}}

\def \cm{{\mathcal{C}}}
\def \grrm{{\mathcal{R}}}
\def \ssm{{\mathcal{SS}}}
\def \ouem{{\mathcal{OUE}}}
\def \suem{{\mathcal{SUE}}}
\def \lhm{{\mathcal{LH}}}
\def \them{{\mathcal{T}}}
\def \Dist{{\mathbb{D}}}
\def \maxrefby {\sqsubseteq^{\mathrm{MAX}}}
\def \refby {\sqsubseteq}

\def \BayesCap{{\mathcal{M}\mathcal{L}}}
\def \gLeak{{\mathcal{L}}_g}
\def \gVuln{\ensuremath{V_g}}

\usepackage{hyperref} 
\def\BibTeX{{\rm B\kern-.05em{\sc i\kern-.025em b}\kern-.08em
    T\kern-.1667em\lower.7ex\hbox{E}\kern-.125emX}}
\begin{document}

\makeatletter
\newcommand{\linebreakand}{%
  \end{@IEEEauthorhalign}
  \hfill\mbox{}\par
  \mbox{}\hfill\begin{@IEEEauthorhalign}
}
\makeatother

\title{Beyond Epsilon: A Principled QIF Framework for Local Differential Privacy}

\author{\IEEEauthorblockN{Ramon G. Gonze\textsuperscript{$\ddagger$}\thanks{\textsuperscript{$\ddagger$}These authors contributed equally to this work.}}
\IEEEauthorblockA{\textit{Institut Polytechnique de Paris} \\
\textit{Universidade Federal de Minas Gerais}\\
Paris, France -- Belo Horizonte, Brazil \\
ramon.gonze@inria.fr}
\and
\IEEEauthorblockN{Natasha Fernandes\textsuperscript{$\ddagger$}}
\IEEEauthorblockA{\textit{Macquarie University} \\
Sydney, NSW Australia \\
natasha.fernandes@mq.edu.au}
\and
\IEEEauthorblockN{Héber H. Arcolezi\textsuperscript{$\P$}\thanks{\textsuperscript{$\P$}Work done while at Inria.}}
\IEEEauthorblockA{\textit{ÉTS Montréal}, Montréal, Canada \\
\textit{Inria, Grenoble, France} \\
heber.hwang-arcolezi@etsmtl.ca}
\linebreakand
\IEEEauthorblockN{Catuscia Palamidessi}
\IEEEauthorblockA{\textit{Inria, Institut Polytechnique de Paris} \\
Paris, France \\
catuscia@lix.polytechnique.fr}
\and
\IEEEauthorblockN{Nataliia Bielova}
\IEEEauthorblockA{\textit{Inria Centre at University Côte d’Azur} \\
Sophia Antipolis, France \\
nataliia.bielova@inria.fr}
}

\maketitle

\begin{abstract}
Local Differential Privacy (LDP) has become the de facto standard for privacy-preserving data collection in large-scale systems, in particular for the purpose of estimating frequencies. However, the current research landscape lacks a systematic and principled way to compare LDP protocols. The parameter $\varepsilon$ of LDP is considered the measure of
privacy, but it only bounds worst-case distinguishability.   Other comparisons rely on utility-driven analyses, where mechanisms are ranked based on their ability to preserve data utility for a given privacy budget $\varepsilon$. Both such kinds of comparisons fail to account for the strength of protocols against diverse attacker models. 
In this paper, we propose a framework for analyzing LDP frequency estimation protocols through the lens of Quantitative Information Flow (QIF). By modeling LDP mechanisms as probabilistic channels, we leverage the concept of refinement (Blackwell ordering) to establish more principled classifications. This approach allows us to determine when one protocol is intrinsically superior to another for all possible adversaries, and to discuss the implications for utility. In particular, our analysis uncovers cases where protocols previously deemed ``optimal'' are, in fact, incomparable with, or strictly dominated by, other protocols. 
We provide a formal QIF-based treatment of seven state-of-the-art LDP protocols, including Generalized Randomized Response (GRR), Subset Selection (SS), local hashing variants (BLH, OLH), unary encoding schemes (SUE, OUE), and Thresholding with Histogram Encoding (THE).  
This perspective
bridges the gap between the LDP and formal methods communities and enables principled, adversary-aware reasoning about
locally private systems.


\end{abstract}


\input{A_text}

\bibliographystyle{ieeetr}
\bibliography{references}

\appendix
\input{A_appendix}

\end{document}

%% file: macros/comments.tex
\ifediting
    \usepackage{comment}
    \usepackage[draft,textsize=footnotesize,textwidth=17mm]{todonotes}
\fi

\newcommand{\openproblem}{\textbf{}:}

\usepackage[strict]{changepage}
\usepackage{framed}
\usepackage{color}

%% file: macros.tex
\newcommand{\pr}{{\sf Pr}} 

\newcommand{\calu}{\mathcal{U}}

\newcommand{\calx}{\mathcal{X}}
\newcommand{\caly}{\mathcal{Y}}

\newcommand{\bbe}{\mathbb{E}}

\newcommand{\Ch}{{\textnormal{\textsf{C}}}}

\definecolor{c1}{HTML}{F2CF96} 
\definecolor{c2}{HTML}{F2CF96} 
\definecolor{c7}{HTML}{E1B0AB}
\definecolor{c9}{HTML}{B3A5B2}
\definecolor{c3}{HTML}{B299C1} 
\definecolor{c5}{HTML}{B299C1} 
\definecolor{c10}{HTML}{A498BE}
\definecolor{c8}{HTML}{8B8BAF}
\definecolor{c4}{HTML}{7E8CB1}

\definecolor{Nc1}{HTML}{E8FFE8} 
\definecolor{Nc5}{HTML}{D8F8D8} 
\definecolor{Nc8}{HTML}{C0EBC0} 
\definecolor{Nc4}{HTML}{A5D8A5} 

\definecolor{NSc1}{HTML}{FFEFEF} 
\definecolor{NSc9}{HTML}{FAE0E0} 
\definecolor{NSc5}{HTML}{F5D5D5} 
\definecolor{NSc10}{HTML}{EFC5C5} 
\definecolor{NSc8}{HTML}{EAB5B5} 
\definecolor{NSc4}{HTML}{D8A0A0} 

\definecolor{NSrc6}{HTML}{E8FFFF} 
\definecolor{NSrc1}{HTML}{CFF9F9} 
\definecolor{NSrc2}{HTML}{CFF9F9} 
\definecolor{NSrc7}{HTML}{BBEFEF} 
\definecolor{NSrc3}{HTML}{AAE3E3} 
\definecolor{NSrc5}{HTML}{AAE3E3} 
\definecolor{NSrc8}{HTML}{99DCDC} 
\definecolor{NSrc4}{HTML}{88D0D0} 

\newcounter{MJ}

\newcounter{CP}

\newcounter{MA}

\newcounter{RG}

%% file: A_text.tex


\section{Introduction}

As the demand for privacy-preserving data collection continues to grow, Local Differential Privacy (LDP)~\cite{first_ldp} has emerged as a robust technique for safeguarding individual privacy. 
Unlike traditional differential privacy~\cite{Dwork2006}, which requires a trusted aggregator, LDP enables users to independently perturb their data before transmission. 
This ensures that sensitive information remains protected, even if the data collector is compromised. 
As a result, LDP has been widely adopted in \emph{large-scale frequency monitoring systems} such as the Google Chrome browser~\cite{rappor}, Apple iOS/macOS~\cite{apple}, Google Gboard~\cite{gboard}, and Windows 10 operating system~\cite{microsoft}.

Frequency estimation protocols~\cite{tianhao2017} constitute a core building block of many LDP deployments.
As a result, their comparative evaluation plays a central role in the design and selection of LDP mechanisms\footnote{“LDP mechanisms” denotes the obfuscation methods, whereas “frequency estimation protocols” refers to these LDP mechanisms together with any postprocessing technique to recover frequencies or histograms from the obfuscated data. Since postprocessing does not weaken privacy guarantees, we often use these two terms interchangeably throughout this paper.}.
However, current evaluation approaches fall short of comparing the strength of the various LDP mechanisms against different kinds of attackers. Indeed, they consider the LDP parameter $\varepsilon$ as the \textit{only} measure of
privacy, and rank protocols on the basis of their utility--in this case, the precision of the frequency estimate, evaluated in terms of error-based metrics--for the same privacy budget $\varepsilon$.
While such classifications are valuable from a statistical perspective, they provide only a partial and utility-dependent view of privacy protection, consequently leading to misleading comparisons of the privacy-utility trade-offs:
As a matter of fact, $\varepsilon$ only bounds the worst-case distinguishability between data points, and two protocols satisfying the same $\varepsilon$-LDP guarantee may exhibit significantly different leakage behaviors under realistic inference attacks, such as data reconstruction~\cite{Gursoy2022} or re-identification~\cite{Arcolezi2023}.
Therefore, from a privacy standpoint, this raises a fundamental question: \emph{What does it mean for one LDP protocol to be strictly more private than another?}

To address this gap, we advocate for a \emph{Quantitative Information Flow} (QIF)~\cite{alvim2020science} perspective on LDP.
QIF provides a principled framework for measuring information leakage of LDP protocols by representing randomized systems as information-theoretic channels and quantifying the attacker’s expected gain for a wide range of threat models.
Crucially, QIF supports \textit{refinement relations}, where one channel is deemed strictly more private than another if it minimizes the potential gain for \emph{any attacker}, regardless of their prior knowledge or specific objectives. This refinement-based approach offers a more robust and comprehensive framework for evaluating protocol superiority than traditional comparisons limited to $\varepsilon$ and utility metrics. As a result, we are able to uncover settings in which protocols commonly regarded as ``optimal'' under error-based metrics are in fact incomparable to, or even strictly dominated by, other protocols.

In this paper, \textbf{we present a formal QIF-based analysis of widely used LDP frequency estimation protocols}, including Generalized Randomized Response (GRR)~\cite{kairouz2016discrete}, Subset Selection (SS)~\cite{Min2018,wang2016mutual}, Binary Local Hashing (BLH)~\cite{Bassily2015}, Optimized Local Hashing (OLH)~\cite{tianhao2017}, Symmetric Unary Encoding (SUE)~\cite{rappor}, Optimized Unary Encoding (OUE)~\cite{tianhao2017}, and Thresholding with Histogram Encoding (THE)~\cite{tianhao2017}.
These mechanisms form the core of many real-world LDP deployments.
For instance, SUE underlies Google's RAPPOR~\cite{rappor} and Microsoft's $d$BitFlipPM~\cite{microsoft} telemetry-based systems, while SS has been recently integrated into Google's Gboard~\cite{gboard}.
Similarly, Apple's Count Mean Sketch mechanism~\cite{apple} relies on data structures closely related to local hashing techniques, such as BLH and OLH.
Finally, THE builds on the widely used Laplace mechanism~\cite{Dwork2006}, a building block of practical differential privacy systems (e.g., LinkedIn API~\cite{linkedin}).

For each of the above protocols, we derive the exact \emph{channel matrix} induced by the local randomization procedure and analyze its leakage properties using standard QIF vulnerability measures.
To support this analysis, we adopt \emph{$f$-differential privacy} ($f$-DP)~\cite{Dong2022}, which characterizes privacy guarantees in terms of 
hypothesis testing, and supports comparison of privacy mechanisms using trade-off functions derived from the most powerful hypothesis tests via the celebrated Neyman-Pearson lemma. This provides a unifying abstraction that aligns naturally with QIF-based refinement~\cite{fernandes2025composition}.
Within this framework, $\varepsilon$-LDP appears as a coarse worst-case bound, while QIF-based refinement relations expose finer-grained distinctions between protocols that $\varepsilon$ alone cannot capture.

\subsection{Contribution}
The main contributions of our paper are as follows:  

\begin{itemize}
    \item We provide a unified QIF-based framework for analyzing and comparing LDP frequency estimation protocols through channel matrices and refinement relations, enabling metric-independent privacy comparisons across adversaries and utility functions. 
    \item We formally bridge the LDP and QIF literatures by showing that data reconstruction attacks studied in prior work~\cite{Gursoy2022,Arcolezi2023} coincide with standard QIF notions of vulnerability.
    In particular, we demonstrate that the expected success of such attacks corresponds to the Bayes vulnerability induced by the LDP channel.
    \item Leveraging this channel-based formulation, we revisit the analysis of Local Hashing (LH) protocols, including BLH and OLH.
    We identify and correct an error in the expected data reconstruction attack analysis reported by~\cite{Gursoy2022}, and derive the correct Bayes risk for LH-based protocols via averaged vulnerability.
    \item We study refinement relations between LDP protocols, uncovering new refinements, which are experimentally corroborated on real-world datasets.

\end{itemize}

\section{Related Work} \label{sec:related_work}

\textbf{Utility-driven LDP protocol design.} A variety of LDP-based frequency estimation protocols have been introduced, all aiming to enhance the privacy-utility trade-off~\cite{tianhao2017,rappor,wang2016mutual,Min2018,Pan2025}. In these works, as well as in most of the differential privacy literature, privacy is characterized by the parameter $\varepsilon$, while utility is evaluated through estimation error, commonly using mean squared error or $l_1$ and $l_2$ distances~\cite{kairouz2016discrete}.
Within this framework, notions of optimality are inherently metric-dependent.
While these analyses provide valuable insights into estimation accuracy, they do not offer a comparison of privacy leakage across protocols, nor do they induce an ordering that is independent of the chosen utility metric.

\textbf{Inference attacks on LDP protocols.} Beyond utility analysis, more recently, a few works~\cite{Gursoy2022,Arcolezi2023,arcolezi2024revealing,Balioglu2026} have analyzed LDP mechanisms through an adversarial lens.
For instance, data reconstruction attacks~\cite{Gursoy2022,arcolezi2024revealing} aim to infer individual user values from locally perturbed reports, while re-identification attacks~\cite{Arcolezi2023} attempt to link anonymised or perturbed records back to individuals.
These attacks reveal that two mechanisms satisfying the same $\varepsilon$-LDP guarantee can exhibit significantly different levels of privacy violations.
However, existing attack analyses typically focus on specific adversarial objectives and threat models, such as maximum-likelihood reconstruction or linkage success rates.
As a result, they provide valuable empirical or analytical evidence of leakage, but remain tied to particular attack strategies.
Consequently, these approaches do not yield a general, attacker-independent ordering of LDP protocols, nor do they explain how different attacks relate to one another from a privacy perspective.


\textbf{Relation between DP/LDP and QIF.} 
There have been several works exploring the connection between differential privacy (DP) and QIF by modeling randomized mechanisms as information-theoretic channels and interpreting privacy as a bound on information leakage. Early results showed that $\varepsilon$-DP implies general, prior-independent upper bounds on min-entropy and Rényi-based leakage measures, thus relating DP guarantees to worst-case QIF notions of leakage \cite{Alvim2011,BartheKopf2011}. This perspective was further developed by explicitly studying DP mechanisms through the dual lenses of information leakage and utility within a unified channel framework \cite{Alvim2015}. More recently, refinement-based approaches from QIF have been applied to DP, providing order-theoretic characterizations of worst-case leakage and enabling systematic comparisons between mechanisms \cite{Chatziko2019}. LDP mechanisms have similarly been analyzed as constrained channels, and QIF measures such as maximal leakage have been used to derive capacity-style interpretations and worst-case leakage bounds \cite{IssaWK20,LopuhaaZwakenbergSkoricLi2019}.
One specific example of an LDP mechanism analyzed using constrained channels is Randomized Response (or GRR), for which the authors in \cite{jurado2023analyzing} measured the privacy amplification achieved via shuffling.
None of the above works, however, uses a QIF approach to compare LDP mechanisms.

\section{Background and Preliminaries}

Throughout the paper, let $[n]=\{1, 2, \ldots, n\}$ denote a set of integers, $\calX$ be the input space with finite domain of size $k = |\calX|$, and $\calY$ the set of (perturbed) outputs of a protocol.

\subsection{Local Differential Privacy}
\label{sub:ldp}

\emph{Local Differential Privacy} (LDP) \cite{first_ldp} formalizes the idea that a single user can protect her data \emph{before} it ever leaves her device.  
Let $\mathcal{X}$ be a finite input domain of size $k \!=\! |\mathcal{X}|$ and let $\mathcal{Y}$ be the (finite) output domain of a randomized mechanism $\mathcal{M}\!:\!\mathcal{X}\!\to\!\mathcal{Y}$.  
We denote the random variables of interest by $X\in\mathcal{X}$ and $Y\!=\!\mathcal{M}(X)\in\mathcal{Y}$.

\begin{definition}[{$\varepsilon$-Local Differential Privacy}]
\label{def:ldp}
$\mathcal{M}$ is said to satisfy $\varepsilon$-local differential privacy ($\varepsilon$-LDP) if, for \emph{every} pair of inputs $x,x'\!\in\!\mathcal{X}$ and for \emph{every} output $y\!\in\!\mathcal{Y}$,
\begin{equation}
  \Pr\!\bigl[\mathcal{M}(x)=y\bigr]
  \;\le\; e^{\varepsilon}\;
          \Pr\!\bigl[\mathcal{M}(x')=y\bigr].
  \label{eq:ldp}
\end{equation}
\end{definition}

Inequality~\eqref{eq:ldp} upper-bounds how much more likely \emph{any} particular output can be under one input than under another.  
The smaller $\varepsilon$ is, the closer the two probabilities are forced to be, and hence the stronger the privacy guarantee.  

\subsection{Quantitative Information Flow (QIF)}
\label{sub:qif}

QIF is a framework for modeling Bayesian attacks against systems represented as information-theoretic channels. These attacks give rise to metrics that quantify how much the channel reveals about a secret.
In our setting, the channel is the LDP mechanism $\mathcal{M}$, whose behavior is fully captured by the row-stochastic \emph{channel matrix}
\[
  C_{x,y}\;=\;\Pr\nolimits_{Y|X}[y\mid x],
  \qquad
  x\in\mathcal{X},\;y\in\mathcal{Y}.
\]
Given a prior distribution $\pi$ over~$\mathcal{X}$, the induced
joint distribution of the \emph{input} $X$ and the \emph{output}
$Y$ is
\begin{equation}
  \Pr\nolimits_{X,Y}[x,y]\;=\;\pi_x\,C_{x,y}.
  \label{eq:joint}
\end{equation}
The marginal of~$Y$ is obtained by summing \eqref{eq:joint} over
all inputs:
\begin{equation}
    \Pr\nolimits_{Y}[y]\;=\;\sum_{x\in\mathcal{X}}\pi_x\,C_{x,y}.
  \label{eq:marginal-y}
\end{equation}
When the prior is uniform, $\pi_x=1/k$, so
$\Pr_{X,Y}[x,y]=C_{x,y}/k$ and \eqref{eq:marginal-y} simplifies to
$\Pr_{Y}[y]=\tfrac1k\sum_x C_{x,y}$.

Our adversarial model is that of a Bayesian adversary with some prior knowledge $\pi:\Dist\calX$ represented as a distribution over secrets $\calX$, a gain function $g: {\cal{W}{\times}\calX} \rightarrow \mathbb{R}$ describing their gain on taking action $w \in \cal{W}$ when the secret has value $x \in \calX$, and who has knowledge of the mechanism $C$. We consider 2 classes of adversaries: the \emph{average-case} adversary, whose success is measured as their \emph{expected} gain after making an observation from the system; and the \emph{max-case} adversary, whose success is measured as the \emph{maximum} gain over all possible observations. Leakage is measured as the (multiplicative) difference between the adversary's prior and posterior knowledge.

Formally, the adversary's prior expected gain is given by
\begin{equation}
  \gVuln(\pi) ~=~ \max_w \sum_x \pi_x g(w, x).
\end{equation}
This is also called the vulnerability of the secret.
The adversary's posterior expected gain is
\begin{equation}
    \gVuln(\pi, C) ~=~ \sum_y \max_w \sum_x \pi_x C_{x,y} g(w,x),
\end{equation}

\noindent with expected leakage
\begin{equation}
    \gLeak(\pi, C) ~=~ \frac{\gVuln(\pi, C)}{\gVuln(\pi)}.
\end{equation}

\noindent In the max-case, the prior gain is~\footnote{Note here we are using the g-leakage version of max-case, which is slightly different from the max-case vulnerability normally defined in QIF, but which we choose for its relevance to the study of LDP.}
\begin{equation}
    \gVuln^{\mathrm{MAX}}(\pi) ~=~ \max_{w,x} \pi_x g(w, x),
\end{equation}

\noindent and the posterior max-case gain is
\begin{equation}
    \gVuln^{\mathrm{MAX}}(\pi, C) ~=~ \max_{y, x, w} \pi_x C_{x, y} g(w, x),
\end{equation}

\noindent with max-case leakage
\begin{equation}
    \gLeak^{\mathrm{MAX}}(\pi, C) ~=~  \frac{\gVuln^{\mathrm{MAX}}(\pi, C)}{\gVuln^{\mathrm{MAX}}(\pi)}.
\end{equation}

We also make use of two robust notions from QIF: \\

1) \textbf{Capacity} which measures the maximum leakage of a system, quantified over all priors and all possible adversaries (within some class). In particular, we note the following.

\begin{itemize}
\item \textbf{Average-case capacity}, also known as \textbf{Bayes capacity} is given by~\cite{alvim2020science}
\begin{equation}\label{def:capacity}
  \BayesCap(C) = \max_{\pi, g} \gLeak(\pi, C) = \sum_y \max_x C_{x, y} \mathrm{.}
\end{equation}

The Bayes capacity is a tight upper bound on all average-case $g$-leakage measures, i.e., it strictly bounds the leakage for any gain function and prior of the adversary (``Miracle Theorem'' \cite{m2012measuring}). One benefit of the Bayes capacity is that it is easy to calculate. In fact, it equals the sum of the maximum values in each column of the channel \cite{BraunCP09}.

\item \textbf{Max-case capacity} is given by~\cite{Fernandes2024}
\begin{equation}
    \BayesCap^{\mathrm{MAX}}(C) = \max_{\pi, g} \gLeak^{\mathrm{MAX}}(\pi, C) = e^\varepsilon,
\end{equation}
where $\varepsilon$ is the minimum LDP guarantee for  channel $C$. 

\noindent The max-case capacity is a tight upper bound on all average-case and max-case $g$-leakage measures~\cite{Fernandes2024}.
\end{itemize}

2) \textbf{Refinement} which is a binary relation and partial order on channels; we say $A$ refines $B$, written $B \refby A$, if $A$ leaks no more than $B$ to any average-case adversary for any prior; equivalently, $B \refby A$ if there exists a channel $W$ such that $B{\cdot}W = A$ where ${\cdot}$ is matrix multiplication. We recall that average-case refinement is stronger than max-case:~\cite{Chatziko2019}
\begin{equation}
  A \refby B \implies A \maxrefby B
\end{equation}
where $\refby^{\mathrm{MAX}}$ is the max-case equivalent refinement.~\footnote{This holds for both the $g$-leakage definition and the traditional definition of max-case.}

We will also make use of a result from \cite{fernandes2025composition} showing that, in the case of binary channels, refinement is equivalent to the $f$-differential privacy ordering between channels.

\paragraph{f-Differential Privacy} Introduced in \cite{Dong2022}, $f$-differential privacy ($f$-DP) is a recent privacy definition allowing more nuanced privacy loss assessments and analyses than traditional differential privacy. $f$-DP models secrets (adjacent datasets) as hypotheses $H_0$, $H_1$, and defines indistinguishability of secrets using the Type-I ($\alpha$) and Type-II ($\beta$) error rates for the most powerful test at various significance levels $\alpha$ (as defined in the Neyman-Pearson lemma). This gives rise to a trade-off function $T_M(\alpha) = \beta$ describing this relationship for a particular mechanism $M$. These trade-off functions can be compared: if $T_M \leq T_N$ pointwise, for mechanisms $M, N$ then $N$ is more private than $M$. It was shown in \cite{fernandes2025composition} that this is exactly equivalent to the QIF notion of refinement. That is,
\begin{equation}
T_M \leq T_N \iff M \refby N ~.
\end{equation}  

In the case of 2x2 channels, the construction of the trade-off function is trivial: by swapping columns if necessary~\footnote{We note that swapping columns does not change the information leakage semantics of the channel.}, the channel $C$ should be put into the form:
\begin{equation}\label{eqn:alpha_beta}
 \begin{bmatrix}
     1 - \alpha & \alpha \\
     \beta & 1 - \beta
 \end{bmatrix}
\end{equation}
where $\beta \leq 1 - \alpha$. Then the trade-off function $T_C$ is the piece-wise linear function joining $(0,1), (\alpha, \beta)$ and $(1, 0)$.

In this paper, we will use trade-off functions to compare channels on 2 secrets under refinement.

\section{Leakage Measures}

In this section, we introduce some leakage measures used in the LDP literature and identify their relationship to QIF measures. This will guide our choice of leakage measures to study in this paper.

\subsection{Attack Model and Leakage for Data Reconstruction}
\label{sub:reconstruction}

As well as reasoning about general safety properties using QIF, we wish to also reason about a particular threat of relevance to our setting: the data reconstruction attack, a one-try attack by an adversary wishing to guess the secret exactly. 

We model the LDP mechanism as a noisy channel with matrix entries
$C_{x,y}\!=\!\Pr_{Y|X}[y\mid x]$. Our adversary makes a single observation $Y = y$ from the mechanism, and outputs a guess $\widehat{x}\in\mathcal{X}$ with the goal of recovering the true input $X = x$.

We equip our adversary with a \emph{uniform} prior $\pi_x=1/ |\calX| = 1/k$.~\footnote{The uniform prior maximizes the attacker’s uncertainty; any non-uniform prior only makes reconstruction easier. This result comes from Bayes capacity, presented in Eq.~\eqref{eq:bayes_capacity}, proposed in~\cite{m2012measuring}.} Bayes’ rule shows that
\[
  \Pr\nolimits_{X|Y}[x\mid y]\;=\;
  \frac{\pi_x\,C_{x,y}}{\Pr_{Y}[y]}
  \;\propto\; C_{x,y},
\]
so the Bayes-optimal (risk-minimizing) decision rule is the \emph{posterior maximum}. That is, 
\begin{equation}
  \widehat{x}(y)\;=\;
  \arg\max_{x\in\mathcal{X}} \Pr\nolimits_{X|Y}[x\mid y]
  \;=\;
  \arg\max_{x\in\mathcal{X}} C_{x,y}.
  \label{eq:posterior-max}
\end{equation}


\noindent\textbf{Remark}: We only consider finite $\calX$ in this paper, so the above (uniform prior, argmax) is well-defined. \\

We can now define the \emph{Adversarial Success Rate}~\cite{Gursoy2022} (or, data reconstruction attack~\cite{arcolezi2024revealing}), which measures the ratio of clients whose true value is correctly predicted by the adversary.


\begin{definition}[Adversarial Success Rate (ASR)]
Let $\widehat{X} = \widehat{x}(Y)$ be the random variable defined by the rule~\eqref{eq:posterior-max}. The Adversarial Success Rate is computed as
\begin{equation}
\begin{split}
  \mathrm{ASR}(C)
  \;&=\;
  \mathbb{E}\left[\Pr\nolimits_{X|Y}[\widehat{X}=X \mid Y]\right] \\
  &= \sum_{y\in\mathcal{Y}}
     \Pr\nolimits_{Y}[y]\,
     \max_{x\in\mathcal{X}}\Pr\nolimits_{X|Y}[x\mid y].
  \label{eq:asr}
\end{split}
\end{equation}
where $C$ is the channel in \eqref{eq:posterior-max}.
\end{definition}


\noindent\textbf{Connection to QIF.} The $0/1$-gain function (also known as Bayes gain function), defined $g(w, x) = 1$ if $w = x$ else $0$ for $w\in\mathcal{W} = \calX$~\footnote{Throughout the paper, we may refer to the vulnerabilities induced by the Bayes gain function as ``1'' instead of explicitly using $g$.}, represents the gain of an adversary whose goal is to guess the secret in 1 try. The posterior vulnerability of a channel $\cm$ under this gain function is given by
\begin{align*}
    V_1(\pi, \cm) ~&=~ \sum_{y \in \calY} \max_w \sum_{x \in \calX} \pi_x \cm_{x,y} g(w,x) \\
    &=~ \sum_{y \in \calY} \max_{x \in \calX} \pi_x \cm_{x,y} \\
    &=~ \sum_{y \in \calY} \max_{x \in \calX} \Pr[y] \frac{\pi_x \cm_{x,y} }{\Pr[y]} \\
    &=~ \sum_{y \in \calY} \Pr[y] \max_{x \in \calX} \Pr[x \mid y]
\end{align*}
which is exactly the ASR given in Eqn~\eqref{eq:asr}.  We therefore have the identity
\begin{equation}
  \mathrm{ASR}(C)\;=\;V_1(\pi, C)~.
\end{equation}
This tells us that the adversarial success rate measures the maximum posterior vulnerability of the system with regard to an attacker who is trying to guess the secret in 1 try from a single observation. Since our attacker is equipped with a \emph{uniform} prior, we then have a direct connection with the average-case capacity. That is,

\begin{equation}\label{eq:bayes_capacity}
    \mathrm{ASR}(C) ~=~ \BayesCap(C) / k~.
\end{equation}

\subsection{Information-Theoretic Leakage Measure}

A number of commonly used leakage measures are derived from information-theoretic quantities and are sometimes employed to assess leakage in DP mechanisms. One such measure is \emph{min-entropy} leakage, which extends Rényi min-entropy to a measure of information leakage and is defined as follows.

\begin{definition}[Min-entropy leakage]
Let $X, Y$ be random variables taking values in finite sets $\calX, \calY$ respectively. Then the min-entropy leakage of X given Y is defined
\begin{equation}
    \mathcal{L}_{\infty}(X|Y) = H_{\infty}(X) - H_{\infty}(X|Y)~,
\end{equation}
where $H_{\infty}(X)\,{=}\log \max_x \Pr_{X}[x]$ and $H_{\infty}(X|Y)\,{=}\,\log \left(\sum_y \Pr_{Y}[y] \max_x \Pr_{X|Y}[x|y]\right)$.
\end{definition}

\noindent \textbf{Connection to QIF.}
%
It was shown in \cite{m2012measuring} that for a channel $C:\calX \rightarrow \calY$ and uniform prior on $\calX$, we have:  
\begin{equation}
   {\mathcal{L}}_{\infty}(X|Y) = \log (\BayesCap(C))~.
\end{equation}


We can transfer the definitions of random variables $X$ and $Y$ to elements of QIF. Let $\calX$ be the set of secrets (space of $X$) and $\pi$ the prior distribution on $\calX$. Let $C$ be a channel that provides the conditional probability $C_{x,y}\,{=}\,\Pr_{Y|X}[y\mid x]$ for the set of outputs $\calY$ (space of $Y$). Setting a uniform prior $\pi_x=1/ |\calX| = 1/k$, and the Bayes gain function, defined as $g(w, x) = 1$ if $w = x$ else $0$ for $w\in\mathcal{W} = \calX$, we have that
\begin{align*}
   {\mathcal{L}}_{\infty}(X|Y) 
   &=  -\log \max_x \Pr\nolimits_{X}[x] \\
   &\hspace{12px} + \log \left(\sum_y \Pr\nolimits_{Y}[y] \max_x \Pr\nolimits_{X|Y}[x|y]\right) \\
   &= -\log V_g(\pi) + \log V_g(\pi, C) \\
   & = -\log (1/k) + \log (\BayesCap(C) / k) \\
   & = \log (\BayesCap(C)).
\end{align*}





\noindent\textbf{Summary of Leakage Measures. } Due to the strong connection between several of the measures identified in this section, we can limit our attention to 2 leakages: 
\begin{enumerate}
    \item the \emph{Bayes capacity} ($\BayesCap$), which upper bounds all average-case leakages; and
    \item the \emph{Max-case capacity} ($\BayesCap^{\mathrm{MAX}}$), corresponding to the $\varepsilon$ of LDP, which upper bounds all average-case and max-case leakages.
\end{enumerate}

We note the relationship, for all channels $C$, that
\[
   \BayesCap^{\mathrm{MAX}}(C) \geq \BayesCap(C).
\]

However, as has previously been shown~\cite{Chatziko2019}, this ordering does not imply an ordering between channels. That is, it can be that $\BayesCap^{\mathrm{MAX}}(A) > \BayesCap^{\mathrm{MAX}}(B)$ but $\BayesCap(A) < \BayesCap(B)$.

\section{LDP Protocols and their Leakages} \label{sec:ldp_protocols}

In this section, we present the $\varepsilon$-LDP protocols under study, focusing on their local perturbation mechanisms.
Proofs can be found in the appendix of the full version of this paper~\cite{full_version}.
While these protocols are typically paired with \emph{unbiased frequency estimators}, our analysis deliberately abstracts away from the estimation step (known in the literature~\cite{tianhao2017}) and concentrates on the information leakage induced by the local randomization.
We derive the corresponding information-channel representations and provide an initial analysis of their leakage properties.
In later sections, these channel representations are used to analyze refinement relations between the mechanisms.



\subsection{Generalized Randomized Response (GRR)} \label{sub:grr}

Generalized Randomized Response (GRR)~\cite{kairouz2016discrete} is a mechanism where each input is reported truthfully with probability $p$ and as any other possible value with probability $1-p$.

\begin{definition}[GRR]\label{def:grr}
Let $\calX = \{ x_1, x_2, \ldots, x_k \}$ and $\calY = \{ y_1, y_2, \ldots, y_k \}$ be finite domains, with $\calY = \calX$ and $|\calX| = k, k \geq 2$.
Then the GRR mechanism for $\varepsilon \geq 0$ is constructed as
\begin{equation} \label{eq:grr}
    \Pr[\mathrm{GRR}(x)=y] = \begin{cases} p=\frac{e^{\varepsilon}}{e^{\varepsilon}+k-1} \textrm{ if } y = x,\\ q=\frac{1}{e^{\varepsilon}+k-1} \textrm{ if } y \neq x \textrm{,} \end{cases}
\end{equation}
\noindent where $y \in \calY$ is the perturbed value sent to the server. 
\end{definition}

\noindent\textbf{Channel Matrix:} The channel matrix ${\grrm}^{\varepsilon}$ for GRR is defined
 \[
    {\grrm}^{\varepsilon}_{ij} = \begin{cases}
    p & \text{if } i = j  \mathrm{,}\\
    q & \text{if } i \ne j \mathrm{.}
    \end{cases}
    \]
    
\noindent \textbf{Example.} For instance, the channel matrix for GRR with $k = 3$ and $\varepsilon = \ln 2$:
    \[
    {\grrm}^{\ln 2} =
    \begin{bmatrix}
    p & q & q \\
    q & p & q \\
    q & q & p \\
    \end{bmatrix} = 
    \begin{bmatrix}
    \nicefrac{1}{2} & \nicefrac{1}{4} & \nicefrac{1}{4} \\
    \nicefrac{1}{4} & \nicefrac{1}{2} & \nicefrac{1}{4} \\
    \nicefrac{1}{4} & \nicefrac{1}{4} & \nicefrac{1}{2} \\
    \end{bmatrix}
     \mathrm{.}
    \]

\noindent\textbf{Leakage Analysis using Bayes capacity:}

Since $p > q$ it is easy to see that every column maximum is $p$ and so we have
\begin{align}
    \BayesCap(\grrm^{\varepsilon}) ~&=~ pk 
    ~=~ \frac{k e^{\varepsilon}}{e^{\varepsilon} + k - 1}
\end{align}

\subsection{Subset Selection (SS)}
Subset Selection (SS)~\cite{Min2018,wang2016mutual} is a protocol that selects a subset of values from the input space and reports it. 
The subset size $\omega$ is determined by $\varepsilon$ and $k$; it is the value that minimizes the variance.


\begin{definition}[SS]
 Let $\calX = \{ x_1, x_2, \ldots, x_k \}$ be a finite domain, let $\varepsilon \geq 0$ and let $\omega = \max \left(1, \left\lfloor \frac{k}{e^{\varepsilon} + 1} \right\rceil \right)$.~\footnote{The notation $\left\lfloor x \right\rceil$ denotes the closest integer value to $x$.} Then, given $x \in \calX$, an output $\mathbf{y}$ of the SS mechanism is constructed as follows:
\begin{enumerate}
    \item Add the true value $x$ to $\mathbf{y}$ with probability $p=\frac{\omega e^{\varepsilon}}{\omega e^{\varepsilon} + k - \omega}$.
    \item If $x$ was added, then $\omega - 1$ values are sampled from $\calX \setminus \{x\}$ uniformly at random (without replacement) and are added to $\mathbf{y}$.
    \item If $x$ was not added, then $\omega$ values are sampled from $\calX \setminus \{x\}$ uniformly at random (without replacement) and are added to $\mathbf{y}$.
\end{enumerate}
The subset $\mathbf{y}$ is then sent to the server.
\end{definition}

\noindent\textbf{Channel Matrix: } The channel matrix $\ssm^{\varepsilon}$ for SS reflects the probability of each input value being included in the reported subset. For input $x_i$ and subset $\mathbf{y}_j$:
    \[
    \ssm^\varepsilon_{i,j} = \begin{cases}
    p \cdot \frac{1}{\binom{k-1}{\omega-1}} & \text{if } x_i \in \mathbf{y}_j \mathrm{,}\\
    (1 - p) \cdot \frac{1}{\binom{k-1}{\omega}} & \text{if } x_i \notin \mathbf{y}_j \mathrm{.}
    \end{cases}
    \]

\noindent \textbf{Example.} For instance, the channel matrix for SS with $k = 3$ and $\omega = 2$ is:


    \[
    \calX = \{0, 1, 2\}, \quad \mathbf{y} \in \binom{\calX}{2} = \{\{0, 1\}, \{0, 2\}, \{1, 2\}\} \mathrm{.}
    \]
    
    \[
    \ssm^{\varepsilon} =
    \begin{bmatrix}
    p \cdot \frac{1}{\binom{2}{1}} & p \cdot \frac{1}{\binom{2}{1}} & (1 - p) \cdot \frac{1}{\binom{2}{2}} \\
    p \cdot \frac{1}{\binom{2}{1}} & (1 - p) \cdot \frac{1}{\binom{2}{2}} & p \cdot \frac{1}{\binom{2}{1}} \\
    (1 - p) \cdot \frac{1}{\binom{2}{2}} & p \cdot \frac{1}{\binom{2}{1}} & p \cdot \frac{1}{\binom{2}{1}} \\
    \end{bmatrix} \mathrm{.}
    \]
    
    Where:
    \[
    \binom{2}{1} = 2, \quad \binom{2}{2} = 1, \quad p = \frac{2 e^{\varepsilon}}{2 e^{\varepsilon} + 3 - 2} \mathrm{.}
    \]

\noindent\textbf{Leakage Analysis using Bayes capacity:}


The Bayes capacity for the SS protocol depends on whether $p{\cdot} \frac{1}{ {k-1 \choose w-1}} > (1-p){\cdot}\frac{1}{{k-1 \choose w}}$. If this is true, then $pk > w$ and we have
    \[
    \BayesCap(C) ~=~ kp{\cdot}{\frac{1}{ {k-1 \choose w-1}}}
    \]
    Otherwise we have
    \[
    \BayesCap(C) ~=~ k(1-p){\cdot} \frac{1}{{k-1 \choose w}}
    \]

\subsection{Local Hashing (LH)}  \label{sub:lh_protocols}

LH protocols~\cite{tianhao2017,Bassily2015} use hash functions to map the input data $x \in \calX$ to a new domain $[g]$, and then obfuscate the hash value with GRR. There are two variations of LH mechanisms: (i) Binary LH (BLH)~\cite{Bassily2015} that just sets $g=2$, and (ii) Optimized LH (OLH)~\cite{tianhao2017} that selects $g=\lfloor \eeps + 1 \rceil$ to minimize the variance.



\begin{definition}[LH]\label{def:lh}
    Let $\mathcal{X} = \{x_1,x_2,\ldots,x_k\}$ be the set of input values and let $\mathscr{H}$ be a universal hash function family such that each $h\in \mathscr{H}$ hashes a value $x\,{\in}\,\mathcal{X}$ into $[g]$, i.e., $h\,{:}\,[k] \rightarrow [g]$, for $k,g \geq 2$. Given an input $x$, an output $y = (y_h,y_p)$ of LH is constructed as follows:
    \begin{enumerate}
        \item \textbf{Encoding step:} A hash function $y_h\,{\in}\,\mathscr{H}$ is chosen uniformly at random, and applied to $x$ to produce the encoded value $y_e\,{=}\,y_h(x)$. 

        \item \textbf{Perturbation step:} Given the encoded value $y_e$, 
        we now perturb it using GRR mechanism, thus producing $GRR(y_e) = y_p$. Finally the output consists of the pair $(y_h,y_p)$ of the hash function 
        $y_h$ and the perturbed encoded value $y_p$.

        \begin{equation}
            \forall i \in [g], \; \Pr[y_p = i] = 
            \begin{cases}
                p = \frac{e^{\varepsilon}}{e^{\varepsilon} + g - 1}, & \text{if } y_e = i \mathrm{,} \\
                q = \frac{1}{e^{\varepsilon} + g - 1}, & \text{if } y_e \neq i \mathrm{.}
            \end{cases}
        \end{equation}
        
    \end{enumerate}
\end{definition}

\noindent\textbf{Channel Matrix:} The channel matrix $\lhm^{\varepsilon}$ for Local Hashing reflects the probability of a pair (the hash function and the perturbed encoded value) being released. The set of inputs is $\mathcal{X} = \{x_1,\ldots,x_k\} = [k]$ and the channel's set of output pairs is $\mathcal{Y}\,{=}\,\{(y_h,y_p)~|~y_h\in\mathscr{H} \text{ and } y_p\in[g]\}$. For input $x$ and pair $y$:

\begin{equation}\label{channel_lh}
    \lhm^\varepsilon_{x,y} = \begin{cases}
    \frac{1}{g^k}\cdot\frac{e^\varepsilon}{e^\varepsilon+g-1} & \text{ , if $y_h(x) = y_p$}\\
    \frac{1}{g^k}\cdot\frac{1}{e^\varepsilon+g-1} & \text{ , if $y_h(x) \neq y_p$}.\\
\end{cases}
\end{equation}

\noindent\textbf{Leakage Analysis using Bayes capacity:}

We start by reasoning about ASR of LH, and following the equivalence of Eqn~\eqref{eq:bayes_capacity}, we obtain $\BayesCap(\lhm^{\varepsilon})$ by multiplying ASR by $k$.

For each reported pair $(y_h, y_p)$, the support set for LH protocols consists of all values $x \in \mathcal{X}$ that hash to $y_p$ and remain equal to $y_p$ after the application of GRR. This support set is denoted by $\setOne_{\mathrm{LH}}= \{x \mid y_h(x) = y_p\}$. Based on the support set derived from each user's report, the adversary can employ one of two possible attack strategies, denoted by $\calA_{\mathrm{LH}}$~\cite{Gursoy2022, Arcolezi2023}.


\begin{itemize}
    \item $\calA^0_{\mathrm{LH}}$ is a random choice $\hat{x}=\mathrm{Uniform}\left( [k] \right)$, if $\setOne_{\mathrm{LH}}=\emptyset$;
    
    \item $\calA^1_{\mathrm{LH}}$ is a random choice $\hat{x}=\mathrm{Uniform}\left( \setOne_{\mathrm{LH}} \right)$, otherwise.
\end{itemize}

The expected ASR of LH protocols was derived in~\cite{Gursoy2022} as
\begin{equation}  \label{eq:exp_asr_lh_wrong}
    \frac{e^{\varepsilon}}{(e^{\varepsilon} + g - 1) \cdot \max \left\{ \frac{k}{g}, 1 \right\}} \mathrm{.}
\end{equation}

\emph{However, this expression does not hold in general.}
In particular, for small domain sizes (e.g., $k=2$), Eq.~\eqref{eq:exp_asr_lh_wrong} did not match our experimental evaluations. 
Using the definition of the LH channel in Eqn~\eqref{channel_lh}, we obtained the following result that matches experimental evaluations for any value of $k$.

\begin{restatable}{proposition}{asrlh}[ASR of Local Hashing]\label{prop:asr_lh}
    Let $\lhm^\varepsilon$ be the channel defined in Eqn~\eqref{channel_lh}. Then
    \begin{equation} \label{eq:exp_asr_lh}
        \text{ASR}(\lhm^\varepsilon) = \frac{e^\varepsilon g^k + (g-1)^k (1-e^\varepsilon)}{(e^\varepsilon + g -1)(kg^{k-1})}.
    \end{equation}
\end{restatable}

In Section~\ref{sec:results}, we present an empirical evaluation comparing Equations~\eqref{eq:exp_asr_lh_wrong} and~\eqref{eq:exp_asr_lh}. 
From Eqn~\eqref{eq:exp_asr_lh}, we derive the Bayes capacity: 

\begin{equation}
    \BayesCap(\lhm^{\varepsilon}) = \frac{e^\varepsilon g^k + (g-1)^k (1-e^\varepsilon)}{(e^\varepsilon + g -1)g^{k-1}}.
\end{equation}

\subsection{Unary Encoding Protocols}

Unary Encoding (UE) protocols encode the input value $i$ as a one-hot $k$-dimensional vector (in which the $i$-th bit is 1) and output a $k$-dimensional bit vector by obfuscating each bit of the input independently.

\begin{definition}[UE]\label{def:ue}
Let $\calX = \{ \mathbf{x}_1, \mathbf{x}_2, \ldots, \mathbf{x}_k \}$, where each $\mathbf{x}_i$ is a one-hot encoded vector of length $k$ in which the bit at position $i$ is set to 1, and the remaining bits are 0. Given an input $\mathbf{x}$, an output vector $\mathbf{y}$ from a Unary Encoding (UE) protocol is constructed as follows: 
\begin{align*}
    \forall{i \in [k]} : \quad \Pr[\mathbf{y}[i]=1] =\begin{cases} p, &\textrm{ if } \mathbf{x}[i]=1 \textrm{,} \\ q, &\textrm{ if } \mathbf{x}[i]=0 \textrm{.}\end{cases}
\end{align*} 
The output $\mathbf{y}$ is sent to the server.
\end{definition}

\noindent\textbf{Channel Matrix:} Each entry \( \cm^{\varepsilon}_{ij} \) in the channel matrix for UE protocols represents the probability of obtaining the output \( \mathbf{y}_j \) given the input \( \mathbf{x}_i \). Formally, for each bit in the input vector \( \mathbf{x}_i \), we can compute:    
    \begin{align*}
     \Pr[\mathbf{y}_j[m] = 1 | \mathbf{x}_i[m] = 1] &= p \\
     \Pr[\mathbf{y}_j[m] = 0 | \mathbf{x}_i[m] = 1] &= 1 - p \\ 
     \Pr[\mathbf{y}_j[m] = 1 | \mathbf{x}_i[m] = 0] &= q \\
     \Pr[\mathbf{y}_j[m] = 0 | \mathbf{x}_i[m] = 0] &= 1 - q
    \end{align*}
    
    \noindent Thus, the probability of obtaining output \( \mathbf{y}_j \) given input \( \mathbf{x}_i \) is:
    \[
    \cm^{\varepsilon}_{ij} = \Pr[\mathbf{y}_j | \mathbf{x}_i] = \prod_{m=1}^k \Pr[\mathbf{y}_j[m] | \mathbf{x}_i[m]]
    \]
    
The following two variations of UE protocols have been studied.


\subsubsection{SUE}
The Symmetric UE protocol (SUE)~\cite{rappor} assigns $p= \frac{e^{\varepsilon/2}}{e^{\varepsilon/2}+1}$ and $q=\frac{1}{e^{\varepsilon/2}+1}$ in Definition~\ref{def:ue} so that $p + q = 1$.

\vspace{0.3cm}
For example, the SUE mechanism for $k=2$ has a channel matrix
 \[
    \setlength{\arraycolsep}{7pt}
    \begin{bNiceMatrix}[first-row, first-col]
    & 00 & 01 & 10 & 11 \\
    01 & pq & p^2 & q^2 & pq \\
    10 & pq & q^2 & p^2 & pq 
    \end{bNiceMatrix}
 \]

\subsubsection{OUE}
The Optimal UE protocol (OUE)~\cite{tianhao2017} was introduced to minimize the variance of the unbiased frequency estimator under $\varepsilon$-LDP among UE protocols.
To achieve this objective, OUE assigns $p=\frac{1}{2}$ and $q=\frac{1}{e^{\varepsilon}+1}$ in Definition~\ref{def:ue}.

For example, the OUE mechanism for $k = 2$ has a channel matrix
 \[
    \setlength{\arraycolsep}{7pt}
    \begin{bNiceMatrix}[first-row, first-col]
    & 00 & 01 & 10 & 11 \\
    01 & (1-p)(1-q) & p(1-q) & q(1-p) & pq \\
    10 &  (1-p)(1-q) & (1-p)q & p(1-q) & pq \\
    \end{bNiceMatrix}
 \]

\vspace{0.5cm}
\noindent\textbf{Leakage Analysis using Bayes capacity for SUE:}

%
%
\noindent Notice that the mechanism for SUE operating on bits has the form
\[
    \setlength{\arraycolsep}{7pt}
    \begin{bNiceMatrix}[first-row, first-col]
      & 0 & 1  \\
    0 & p & q \\
    1 & q  & p \\
    \end{bNiceMatrix}
 \]
Thus, given that $p > q$, the Bayes capacity for SUE can be computed using the closest input one-hot vector (in Manhattan distance) to the output vector. This yields

\begin{itemize}
    \item $\mathbf{y}$ all 0's: maximum posterior is $p^{k-1}q$
    \item $\mathbf{y}$ containing $i$ 1's and $k-i$ 0's: maximum posterior is $p^{k-i+1} q^{i-1}$ and there are $\binom{k}{i}$ such outputs
\end{itemize}

Summing these gives:
\begin{align*}
    \BayesCap(\suem^{\varepsilon}) ~&=~ (\sum_{i = 1}^k \binom{k}{i} p^{k-i+1} q^{i-1} ) + p^{k-1}q \\
    &=~ (\sum_{i=1}^k \binom{k}{i} p^{k-i}q^i{\times}\frac{p}{q}) + p^{k-1}q \\
    &=~ (\sum_{i=0}^k \binom{k}{i} p^{k-i}q^i{\times}\frac{p}{q}) + p^{k-1}q - \binom{k}{0} p^{k} {\times} \frac{p}{q} \\
    &=~ \frac{p}{q} (p+q)^k + p^{k-1}q - p^{k+1}q^{-1} \\
    &=~ \frac{p}{q} (p+q)^k + p^k (\frac{q}{p} - \frac{p}{q})
\end{align*}

Recalling that $p + q = 1$ this simplifies to 
\begin{align*}
   \BayesCap(\suem^{\varepsilon}) &= \frac{p}{1-p} + p^k((1-p)^2 - p^2) / (p(1-p)) \\
    &= (p + p^{k-1}(1 - 2p) )/ (1-p)
\end{align*}

\noindent\textbf{Bayes capacity for OUE:} The mechanism for OUE operating on bits has the form
\[
    \setlength{\arraycolsep}{7pt}
    \begin{bNiceMatrix}[first-row, first-col]
      & 0 & 1  \\
    0 & 1-q & q \\
    1 & 1/2  & 1/2 \\
    \end{bNiceMatrix}
\]

Given $1 - q \geq 1/2$, we can again compute the Bayes capacity using the closest input vector to the observed output. This yields:
\begin{itemize}
\item $\mathbf{y}$ containing all 0's: maximum posterior is $\frac{1}{2} (1-q)^{k-1}$
\item $\mathbf{y}$ containing $i$ 1's and $k-i$ 0's: maximum posterior is $\frac{1}{2} q^{i-1} (1-q)^{k-i}$ and there are $\binom{k}{i}$ such outputs.
\end{itemize}

Summing these yields:
\begin{align*}
    \BayesCap(\ouem^{\varepsilon}) ~&=~ \frac{1}{2} (1-q)^{k-1} + \frac{1}{2} \sum_{i = 1}^k \binom{k}{i} q^{i-1} (1-q)^{k-i} \\
    &=~ \frac{1}{2} (1-q)^{k-1} + \frac{1}{2q} \sum_{i = 0}^k \binom{k}{i} q^{i} (1-q)^{k-i} \\ 
    & \qquad \qquad - \frac{1}{2q} (1-q)^k \\
    &=~ \frac{1}{2q} \Bigl( q(1-q)^{k-1} + 1 - (1-q)^k \Bigr) \\
    &=~ \frac{1}{2q} \Bigl( (1-q)^{k-1} (2q - 1) + 1 \Bigr).
\end{align*}

\subsection{Thresholding with Histogram Encoding (THE)}

Like the UE protocols, Thresholding with Histogram Encoding (THE)~\cite{tianhao2017} encodes the input value $i$ as a one-hot $k$-dimensional vector, $\mathbf{x}=[0, 0, \ldots, 1, 0, \ldots, 0]$, in which only the $i$-th component is $1$. It outputs a $k$-dimensional bit vector by independently perturbing each input bit.

\begin{definition}[THE]
Let $\calX = \{ \mathbf{x}_1, \mathbf{x}_2, \ldots, \mathbf{x}_k \}$, where each $\mathbf{x}_i$ is a one-hot encoded vector of length $k$ in which the bit at position $i$ is set to 1, and the remaining bits are 0. Given an input $\mathbf{x}$ and a threshold $\theta \in (0.5,1)$, an output vector $\mathbf{y}$ from a THE protocol is constructed as follows:
\begin{enumerate}
    \item For each $i \in [k]$, $\mathbf{y}[i] := \mathbf{x}[i] + Z$, where $Z \sim \mathrm{Lap}(\frac{2}{\varepsilon})$ is a random variable drawn from a Laplace distribution with mean 0 and scale $\frac{2}{\varepsilon}$.
    \item For each $i \in [k]$, $\mathbf{y}[i] = \begin{cases}
1, & \text{if } \mathbf{y}[i] > \theta \\
0, & \text{if } \mathbf{y}[i] \leq \theta
\end{cases}$
\end{enumerate}
\end{definition}

The resulting output vector $\mathbf{y}$ is a binary vector in $\{0, 1\}^k$, which satisfies the following for all $i \in [k]$~\cite{tianhao2017}:
\begin{align*}
    \Pr[\mathbf{y}[i] = 1 \mid \mathbf{x}[i] = 1] &= 1 - \frac{1}{2} e^{\frac{\varepsilon (\theta-1)}{2}}  &= p\\
    \Pr[\mathbf{y}[i] = 1 \mid \mathbf{x}[i] = 0] &= \frac{1}{2} e^{-\frac{\varepsilon \theta}{2}}  &= q
\end{align*}

\noindent \textbf{Channel Matrix:} 
    Each entry \(\them^{\varepsilon}_{ij}\) of the THE channel matrix represents the probability of obtaining the output vector \(\mathbf{y}_j\) given the input vector \(\mathbf{x}_i\) and is given by:
    
    \[
    \them^{\varepsilon}_{ij} = \prod_{m=1}^{k} \Pr(\mathbf{y}_j[m] \mid \mathbf{x}_i[m]) \mathrm{.}
    \]

\noindent \textbf{Example} The channel matrix of THE for $k = 2$ is:

    \begin{equation*}
        \them^{\varepsilon} = 
        \begin{bmatrix}
            (1-p)(1-q) & (1-p)q & p(1-q) & pq \\
            (1-p)(1-q) & p(1-q) & (1-p)q & pq \\
        \end{bmatrix}
    \end{equation*}

\noindent\textbf{Leakage Analysis using Bayes capacity:}

Observe that the mechanism operating on bits has the form
\[
    \begin{bmatrix}
        1-q & q \\
        1-p & p
    \end{bmatrix}
\]

\noindent where $p > q$. As with the UE mechanisms, we can compute the Bayes capacity using the closest input one-hot vector to the observed (output) vector. This gives
\begin{itemize}
    \item $\mathbf{y}$ all 0's: maximum posterior is $(1-p)(1-q)^{k-1}$
    \item $\mathbf{y}$ containing $i$ 1's and $k-i$ 0's: maximum posterior is $p q^{i-1} (1-q)^{k-i}$ and there are $\binom{k}{i}$ such outputs.
\end{itemize}

Summing these gives:
\begin{align*}
    \BayesCap(\them^{\varepsilon}) ~&=~ (1-p)(1-q)^{k-1} + p \sum_{i = 1}^k \binom{k}{i} q^{i-1} (1-q)^{k-i} \\
    &=~ (1-p)(1-q)^{k-1} + \frac{p}{q} \sum_{i=1}^k \binom{k}{i} q^i (1-q)^{k-i} \\
    &=~ (1-p)(1-q)^{k-1} + \frac{p}{q} - \frac{p}{q}(1-q)^k \\
    &=~ (1-q)^{k-1}\left(1 - \frac{p}{q}\right) + \frac{p}{q}.
\end{align*}

\section{Formal Analysis of LDP Protocols}
All proofs for this section can be found in the appendix of the full version of this paper~\cite{full_version}.

\subsection{Refinement Families}\label{sec:families}


In this section, we study the relationship between LDP protocols represented as channels from the perspective of refinement. Recall that the refinement relation $\refby$ tells us whether one channel is safer than another against \emph{any} adversary (average or max-case), regardless of their prior knowledge. We begin by considering channels within the same ``family''.

\begin{definition}[$\varepsilon$-family]
Let $M_\varepsilon: \calX \rightarrow \calY$ be a mechanism parametrized by $\varepsilon$ so that for every $\varepsilon \geq 0$
we have that $M_\varepsilon$ is $\varepsilon$-DP. We call the set of mechanisms $M_\varepsilon$ an $\varepsilon$-family and denote it by $\mathbb{F}_M$.
\end{definition}

\begin{definition}[Refinement family]\label{def:ref_family}
Let $\mathbb{F}_M$ be an $\varepsilon$-family of mechanisms. We call $\mathbb{F}_M$ a \emph{refinement family} if
\[ 
   \varepsilon \geq \varepsilon' \iff M_\varepsilon \refby M_{\varepsilon'}
\]
\end{definition}

We note that for Def~\ref{def:ref_family} to hold, it is sufficient to show the forward direction only, since the reverse direction is already known to hold in general~\cite{Chatziko2019}.

Refinement families are important for managing the privacy-utility trade-off.

\begin{restatable}{theorem}{monotonicity}[Monotonicity of utility]
Let $\mathbb{F}_M$ be a refinement family of mechanisms. Let $\mathcal{U}(M)$ be a utility function that measures the expected loss of the mechanism $M$ to an analyst who observes a value $y$ and makes an optimal guess $\hat{y}$. Then the utility function $\mathcal{U}(M)$ is monotonic on $\varepsilon$. That is, 
\[
   \varepsilon \geq \varepsilon' \iff \mathcal{U}(M_\varepsilon) \geq \mathcal{U}(M_{\varepsilon'})~.
\]
\end{restatable}
\begin{proof}
    This follows directly from refinement, since $A \refby B \iff \mathcal{U}(A) \geq \mathcal{U}(B)$ for all utility functions $\mathcal{U}$ expressible using loss functions under expectation.
\end{proof}

We also recall Eqn~\eqref{eqn:alpha_beta} describing the 2x2 channel in terms of $\alpha$ and $\beta$, which we use in the following manner:

\begin{definition}\label{def:alpha_beta}
Let $C$ be a 2x2 channel written as:
\[
   C = \begin{bmatrix}
       1 - \alpha & \alpha \\
       \beta & 1 - \beta
   \end{bmatrix}
\]    
where $\beta \leq 1 - \alpha$. Then the trade-off function $T_C$ is characterized by the point $(\alpha, \beta)$ which we call the \emph{trade-off point} for $C$.
\end{definition}

We now show how trade-off points relate to refinement.

\begin{restatable}{lemma}{lemmaequiv}\label{lem:equiv}
    Let $C, D$ be 2x2 channels with trade-off points ($\alpha_C, \beta_C$) and ($\alpha_D, \beta_D$) respectively. Then the following statements are equivalent:
    \begin{enumerate}
        \item $C \refby D$
        \item The posteriors of $D$ under a uniform prior are inside the convex hull of the posteriors of $C$
        \item $\frac{\beta_C}{1 - \alpha_C} \leq \frac{\beta_D}{1 - \alpha_D}$ and $\frac{1 - \beta_C}{\alpha_C} \geq \frac{1 - \beta_D}{\alpha_D}$.
    \end{enumerate}
\end{restatable}

The following sufficiency condition will simplify some of the proofs.

\begin{restatable}{lemma}{sufficiency}\label{lem:sufficiency}
Let $C, D$ be 2x2 channels with trade-off points ($\alpha_C, \beta_C$) and ($\alpha_D, \beta_D$) respectively. Then $\alpha_C \leq \alpha_D$ and $\beta_C \leq \beta_D$ implies that $C \refby D$.
\end{restatable}

We are now ready to study the refinement properties of the LDP protocols considered in this paper.

\subsubsection{The GRR Mechanism}

The refinement order of the GRR mechanism parametrized by $\varepsilon$ was studied in \cite{Chatziko2019} in which the following was shown:

\begin{lemma}~\cite{Chatziko2019}
    The family $\mathbb{F}_{\mathrm{GRR}}$ of GRR mechanisms is a refinement family.
\end{lemma}

\subsubsection{UE and THE Mechanisms}

These mechanisms are derived from 2x2 mechanisms operating on bits. The following construction will be helpful in proving some properties of these mechanisms.

\begin{enumerate}
\item Start with a 2x2 mechanism $B$ which takes a bit in $\{0, 1\}$ to a bit in $\{0, 1\}$. 
\item Construct the $k$-wise Kronecker product $B^{\otimes k}$ which takes a $k$-bit vector to a $k$-bit vector. Note that the domain of $B^{\otimes k}$ is the set of all $k$-bit vectors.
\item Delete from $B^{\otimes k}$ all rows other than those that correspond to one-hot vectors. Call this mechanism $B^{\otimes k}_{\textrm{hot}}$.
\end{enumerate}

Notice that $B^{\otimes k}_{\textrm{hot}}$ is a channel. The above construction has some nice properties with respect to refinement.

\begin{restatable}{lemma}{kronrefinement}\label{lem:kron_refinement}
    Let $A, B$ be bitwise channels and let $A \refby B$. Then $A^{\otimes k} \refby B^{\otimes k}$ and $A^{\otimes k}_{\mathrm{hot}} \refby B^{\otimes k}_{\mathrm{hot}}$.
\end{restatable}

This tells us that it is sufficient to prove refinement on bitwise channels for the UE and THE mechanisms in order to have refinement on the overall channels.

We first consider the set of THE mechanisms. Recall that the THE bitwise mechanism is defined

\[
   T_{\varepsilon, \theta} ~=~ \begin{bmatrix}
       1-\frac{1}{2} e^{\frac{-\varepsilon \theta}{2}} & \frac{1}{2} e^{\frac{-\varepsilon \theta}{2}} \\
       \frac{1}{2}e^{\frac{\varepsilon (\theta - 1)}{2}} & 1 - \frac{1}{2}e^{\frac{\varepsilon (\theta - 1)}{2}}
   \end{bmatrix}
\]

\noindent for $\theta \in (0.5, 1)$. We have the following:

\begin{restatable}{lemma}{neyman}\label{lem:neyman}
Let $\theta \in (0.5,1)$, and let $\varepsilon, \varepsilon' \geq 0$ with $\varepsilon \leq \varepsilon'$. Then $T_{\varepsilon', \theta} \refby T_{\varepsilon, \theta}$.
\end{restatable}

We give an example of such THE mechanisms in Figure \ref{fig:trade-off}. This now gives our main result for THE mechanisms.

\begin{figure}
    \centering
 \includegraphics[width=0.5\linewidth]{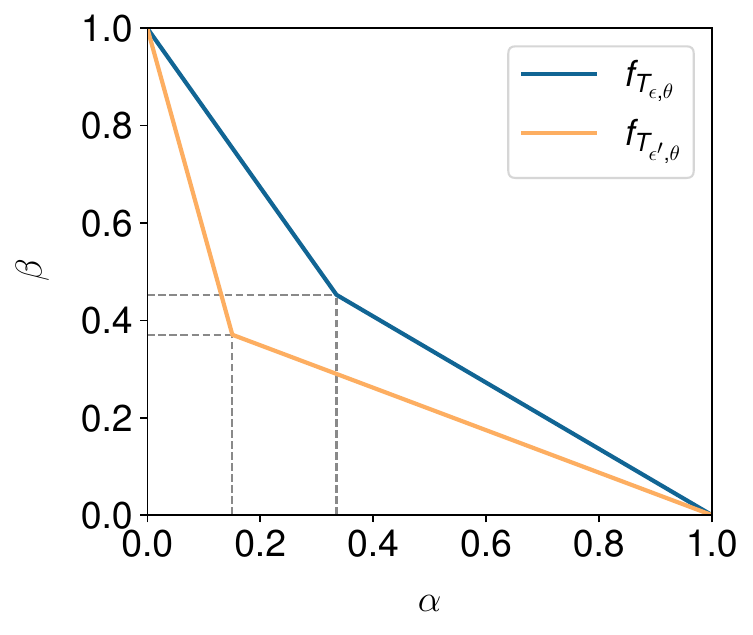}
    \caption{Example showing ordering between $f$-privacy trade-off functions for THE mechanisms when $\alpha, \beta$ values are ordered, i.e. $\alpha > \alpha'$ and $\beta > \beta'$.}\label{fig:trade-off}
\end{figure}

\begin{lemma}
The family $\mathbb{F}_{\mathrm{THE}}$ of THE mechanisms is a refinement family.
\begin{proof}
   Immediate from Lemma \ref{lem:neyman} and Lemma~\ref{lem:kron_refinement}.
\end{proof}
\end{lemma}

Similarly, we find that SUE mechanisms also form a refinement family.

\begin{lemma}
    The family $\mathbb{F}_{\mathrm{SUE}}$ of SUE mechanisms is a refinement family.
\begin{proof}
The proof follows similarly to the proof for THE mechanisms given in Lemma \ref{lem:neyman}.
Consider first the bitwise mechanism for SUE, given by:
\[
    S_{\varepsilon} = \begin{bmatrix}
        \frac{e^{\varepsilon/2}}{e^{\varepsilon/2}+1} & \frac{1}{e^{\varepsilon/2}+1} \\
        \frac{1}{e^{\varepsilon/2}+1} & \frac{e^{\varepsilon/2}}{e^{\varepsilon/2}+1}
    \end{bmatrix}
\]
Observe that $e^{\varepsilon/2} \geq 1$ and so, using the Neyman-Pearson lemma, we have that the most powerful test at significance $\frac{1}{e^{\varepsilon/2}+1}$ has power $\frac{e^{\varepsilon/2}}{e^{\varepsilon/2}+1}$. That is, we have $\alpha = \beta = \frac{1}{e^{\varepsilon/2}+1}$.

Next, if $\varepsilon \leq \varepsilon'$ then $\alpha \geq \alpha'$ and $\beta \geq \beta'$ which gives a pointwise ordering on trade-off functions for $S_\varepsilon, S_{\varepsilon'}$. This implies that $S_{\varepsilon'} \refby S_{\varepsilon}$.
The result follows from Lemma~\ref{lem:kron_refinement}.
\end{proof}
\end{lemma}

Finally, for the OUE mechanisms:

\begin{lemma}
    The family $\mathbb{F}_{\mathrm{OUE}}$ of OUE mechanisms is a refinement family.
\begin{proof}
    As for the THE and SUE mechanisms, consider the bitwise OUE mechanism given by
    \[
    \begin{bmatrix}
        \frac{e^{\varepsilon}}{e^{\varepsilon}+1} & \frac{1}{e^{\varepsilon}+1} \\
        1/2 & 1/2
    \end{bmatrix}
    \]
Notice that $\frac{e^{\varepsilon}}{e^{\varepsilon}+1} \geq 1/2$ and therefore from Neyman-Pearson we have $\alpha = \frac{1}{e^{\varepsilon}+1}$ and $\beta = 1/2$.
Now, given $\varepsilon \leq \varepsilon'$ we find that $\alpha \geq \alpha'$ but $\beta = \beta' = 1/2$. This  gives an ordering on trade-off functions so that $f_{O_\varepsilon} \geq f_{O_\varepsilon'}$ (from Lemma \ref{lem:sufficiency}) and thus $O_{\varepsilon'} \refby O_{\varepsilon}$. The result follows from Lemma~\ref{lem:kron_refinement}.
\end{proof}
\end{lemma}

\subsubsection{SS and LH Mechanisms} We did not find an obvious refinement relation for these mechanisms, so we leave a more in-depth study of them to future work.

\subsection{Comparing Mechanisms}\label{sec:comparing}

In this section, we study the relationship between the different refinement families identified in the previous section.


Firstly, notice that the UE and THE mechanisms have the same form: they perturb each bit of the secret independently, using a mechanism of the form:
\[
   B ~=~ \begin{bmatrix}
       (1-q) & q \\
       (1 - p) & p
   \end{bmatrix}
\]

The SUE (symmetric UE) protocol chooses $p = \frac{e^{\varepsilon/2}}{e^{\varepsilon/2} + 1}$ and $q = 1-p$. So the SUE channel is $\frac{\varepsilon}{2}$-DP. Interestingly, it is \emph{universally optimal} among the $\frac{\varepsilon}{2}$-DP mechanisms on 2 secrets (as shown in \cite{FernandesCSF22}). This means that it leaks the most information compared with any other mechanism; in QIF terms, it is the unique minimal element in the \emph{anti-refinement} chain for $\frac{\varepsilon}{2}$-DP mechanisms.

However, this unique optimality does not transfer through to the SUE mechanism defined on one-hot vectors for two reasons: firstly, (as shown in \cite{FernandesCSF22}), there are no unique optimal DP mechanisms on $>2$ secrets (where refinement is not a lattice); and secondly, because the other bitwise mechanisms we study here (OUE, THE) are not $\frac{\varepsilon}{2}$-DP on bits (and therefore the result of \cite{FernandesCSF22} does not apply). 

On the other hand, refinement between bitwise mechanisms \emph{does} transfer through to refinement on one-hot vector domains (Lemma~\ref{lem:kron_refinement}). Thus, we will continue by examining refinement relationships between the bitwise mechanisms, noting that this is sufficient but not necessary for refinement on the full mechanism. 


\vspace{0.2cm}
\noindent\textbf{OUE and SUE mechanisms}


We first remark (continuing from above) that the bitwise OUE mechanism has a different DP guarantee than the bitwise SUE mechanism. This is why we cannot use a simple ``universal optimality'' argument to show that the SUE is always better (for utility) than the OUE. In fact, we have the following:

\begin{restatable}{lemma}{ouesue}\label{lem:ouesue}
Let $S_{\varepsilon}, O_{\varepsilon}$ be the bitwise mechanisms for the SUE and OUE protocols, respectively, both parametrized by some $\varepsilon > 0$. Then $S_{\varepsilon} \not\refby O_{\varepsilon}$ and $O_{\varepsilon} \not\refby S_{\varepsilon}$.
\end{restatable}
\begin{proof}[Proof (Sketch)]
    We use the $f$-privacy trade-off functions to compare $O_{\varepsilon}$ and $S_\varepsilon$, this time using Lemma \ref{lem:equiv} since the condition for Lemma \ref{lem:sufficiency} does not hold. The result follows.
\end{proof}

Note that this only means that neither \emph{bitwise} mechanism is better than the other (in terms of utility). Importantly, this does not mean that one mechanism is never better than the other for any utility measure; failure of refinement here means we cannot prove that for all utility measures one mechanism is universally better than the other.~\footnote{Also note that the comparison on bitwise mechanisms is only a sufficiency condition for refinement.}

\vspace{0.2cm}
\noindent\textbf{THE and OUE mechanisms}

\begin{figure}[!h]
    \centering
    \includegraphics[width=0.48\linewidth]{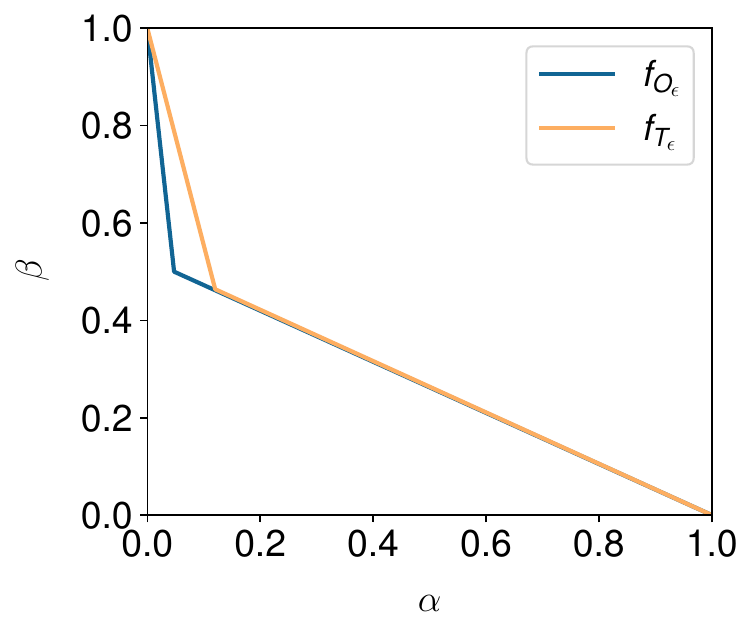}   \includegraphics[width=0.48\linewidth]{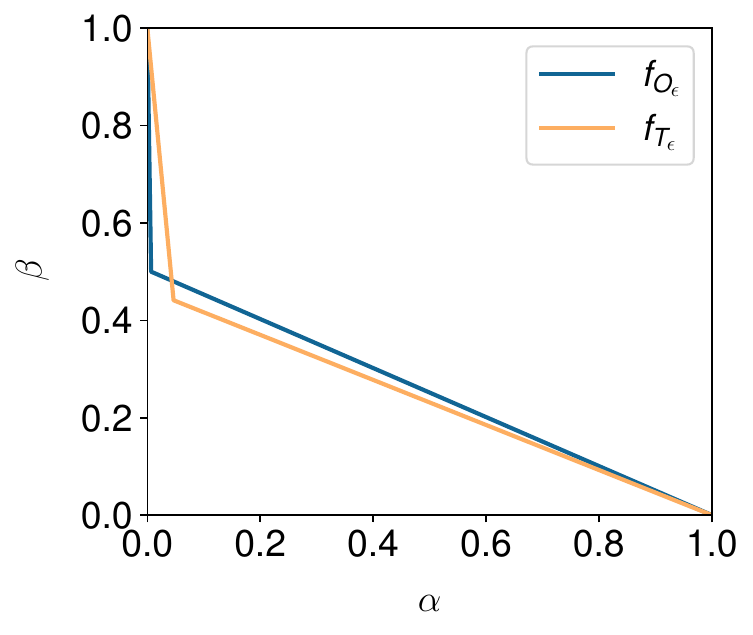}
    \caption{Refinement relation between OUE and THE mechanisms for $\theta = 0.95$. The left figure shows refinement for $\varepsilon = 3$ and the right shows non-refinement for $\varepsilon = 5$.}
    \label{fig:OUE_THE_relation}
\end{figure}

Recall that the THE mechanism is parametrized by both $\varepsilon$ and $\theta \in (1/2, 1)$. We find that the value of $\theta$ plays a role in the refinement relation between THE and OUE, as given in the following.

\begin{restatable}{lemma}{ouethe}\label{lem:theta_ineq}
Let $T_{\varepsilon, \theta}, O_{\varepsilon}$ be bitwise mechanisms for THE, OUE, respectively. Then $O_\varepsilon \refby T_{\varepsilon, \theta}$ exactly when 
\[
\theta \geq \frac{2 \ln(e^{\varepsilon} + 1 + (e^{\varepsilon/2} - 1)(\sqrt{e^\varepsilon + 1})) - \varepsilon - 2\ln 2}{\varepsilon}
\]
for $\varepsilon > 0$ and $\theta \geq 1/\sqrt{2}$ at $\varepsilon = 0$.
\end{restatable}
\begin{proof}{(Sketch)}
We make use of Lemma \ref{lem:equiv} (3). The second inequality always holds, but the first inequality only holds for the value of $\theta$ shown above.
Since it is undefined at $\varepsilon = 0$, we can take limits to get the value of $1/\sqrt{2}$.
\end{proof}

To give some intuition for what this means, we remark that for $\theta = 0.5$, refinement never holds, and for $\theta = 1$, refinement always holds. In between these values (actually above $\theta = 1/\sqrt{2}$), $\theta$ is an increasing function of $\varepsilon$, and so refinement holds for small $\varepsilon$ values but fails to hold when $\varepsilon$ goes above a threshold. 

For example, setting $\varepsilon = 0.8$ in the above inequality produces $\theta \geq 0.8$. Conversely, this means that for $\theta = 0.8$, refinement holds for $\varepsilon \leq 0.8$ and breaks for $\varepsilon > 0.8$. An example is depicted in Figure \ref{fig:OUE_THE_relation}.


From Lemma \ref{lem:kron_refinement} it follows that the corresponding one-hot mechanisms are also in refinement when $\theta$ satisfies the above.

\begin{corollary}
    Let $T_{\varepsilon, \theta}, O_{\varepsilon}$ be bitwise mechanisms for THE, OUE respectively and let $\theta$ satisfy the inequality in Lemma~\ref{lem:theta_ineq}. Then the corresponding one-hot mechanisms $T^{\otimes k}_{\mathrm{hot}}, O^{\otimes k}_{\mathrm{hot}}$ satisfy $O^{\otimes k}_{\mathrm{hot}} \refby T^{\otimes k}_{\mathrm{hot}}$.
\end{corollary}

\vspace{0.2cm}
\noindent\textbf{THE and SUE mechanisms}

Here we find that there is always a refinement, regardless of the value taken by $\theta$ for the THE mechanism. 

\begin{restatable}{lemma}{thesue}\label{lem:sue_ref_the}
Let $T_{\varepsilon, \theta}, S_{\varepsilon}$ be bitwise mechanisms for THE, SUE, respectively. Then  $S_{\varepsilon} \refby T_{\varepsilon, \theta}$ for any $\theta \in (0.5,1)$.
\end{restatable}

\noindent\textbf{Remarks} On 2 secrets, refinement is a lattice~\cite{fernandes2025composition}, meaning that there is a minimal element that is a refinement of both OUE and SUE mechanisms, and a maximal element that is an anti-refinement of both OUE and SUE.

\begin{figure}[h]
    \centering
 \includegraphics[width=0.5\linewidth]{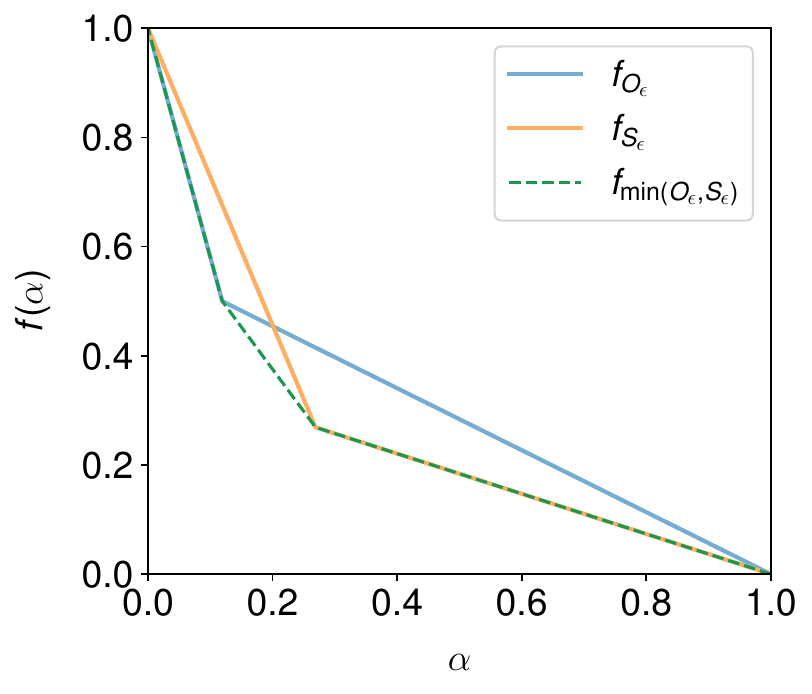}
    \caption{OUE and SUE mechanisms are not in refinement, but there is an element given by min(O, S) (green dotted line) which is the maximum element refined by both OUE and SUE.}\label{fig:min}
\end{figure}

An example of the latter is given in Figure \ref{fig:min}. The intuition behind the min of SUE and OUE is that it is the most private mechanism that is at least as useful as SUE and OUE (for all utility functions). The construction of the min in general has been shown~\cite{fernandes2025composition}, and we formulate it here as follows.

\begin{lemma}\label{lem:oue_min_sue}
Let $S_\varepsilon, O_\varepsilon$ be bitwise SUE and OUE mechanisms, respectively, parametrized by $\varepsilon \geq 0$. Then the mechanism $M_\varepsilon = S_\varepsilon \min O_\varepsilon$, defined
\[
      M_{\varepsilon} = \begin{bmatrix}
         \frac{e^{\varepsilon/2}}{e^{\varepsilon/2} + 1} &  \frac{1}{e^{\varepsilon/2} + 1} - \frac{1}{e^\varepsilon + 1} &  \frac{1}{e^{\varepsilon} + 1} \\
         \frac{1}{e^{\varepsilon/2} + 1} & \frac{e^{\varepsilon/2}}{e^{\varepsilon/2} + 1}  - \frac{1}{2} & \frac{1}{2} 
      \end{bmatrix},
\]
is the maximum element in the refinement order that is an anti-refinement of both $S_\varepsilon$ and $O_\varepsilon$.
\end{lemma}

We leave the study of this mechanism to future work.

\vspace{0.2cm}
\noindent\textbf{Max-case refinement}

We may wonder if average-case refinement is too strong, and max-case refinement could be a better measure of security, especially given that $\varepsilon$ describes a worst-case attack. In fact, the following result is immediate from Lemma \ref{lem:equiv}.
\begin{corollary}\label{cor:coincide}
On 2x2 channels, average-case refinement and max-case refinement coincide.~\footnote{Here we use the condition for traditional max-case refinement, noting that it subsumes max-case for the $g$-leakage version used in this paper.}
\begin{proof}
Lemma \ref{lem:equiv} (2) is exactly the condition for max-case refinement, and we have shown that on 2x2 channels it coincides with average-case refinement.
\end{proof}
\end{corollary}

From this, it is also immediate that max-case refinement is preserved for the full mechanisms (defined on one-hot vectors).

\begin{corollary}
Let $C, D$ be 2x2 channels with $C \maxrefby D$. Then $C^{\otimes k} \maxrefby D^{\otimes k}$ and $C^{\otimes k}_{\mathrm{hot}} \maxrefby D^{\otimes k}_{\mathrm{hot}}$.
\begin{proof}
From Cor. \ref{cor:coincide}, $C \maxrefby D \implies C \refby D$, and from Lemma \ref{lem:kron_refinement} we have $C \refby D \implies C^{\otimes k}_{\mathrm{hot}} \refby D^{\otimes k}_{\mathrm{hot}}$. Finally, we know that average-case is strictly stronger than max-case in general (\cite{Chatziko2019}) and so $C^{\otimes k}_{\mathrm{hot}} \refby D^{\otimes k}_{\mathrm{hot}} \implies C^{\otimes k}_{\mathrm{hot}} \maxrefby D^{\otimes k}_{\mathrm{hot}}$.
\end{proof}
\end{corollary}

\section{Experimental Results} \label{sec:results}

In order to present an experimental evaluation of the theoretical results described in this work, we selected the Kosarak dataset \cite{kosarak-dataset}, a click-stream dataset, i.e., sequences of actions (clicks) performed by users while browsing a website, from a Hungarian news portal. The dataset contains 41,270 distinct actions and approximately 8 million actions performed by 990,002 users.

Due to computational constraints, we considered two subsets of the data. The first subset, used in Figure~\ref{fig:privacy_the_ue}, consists of 10,000 actions drawn from the 10 most frequent actions. The second subset, used in Figure~\ref{fig:utility_the_ue}, consists of 10,000 actions selected uniformly at random. The code used in the experiments is available on GitHub~\footnote{\url{https://github.com/ramongonze/ldp_qif}.}.

\subsection{Refinement of THE, OUE and SUE}

We start our analysis with a comparative study of information leakage among the THE, OUE, and SUE protocols.

\begin{figure}[h]
    \centering
    \includegraphics[width=1\linewidth]{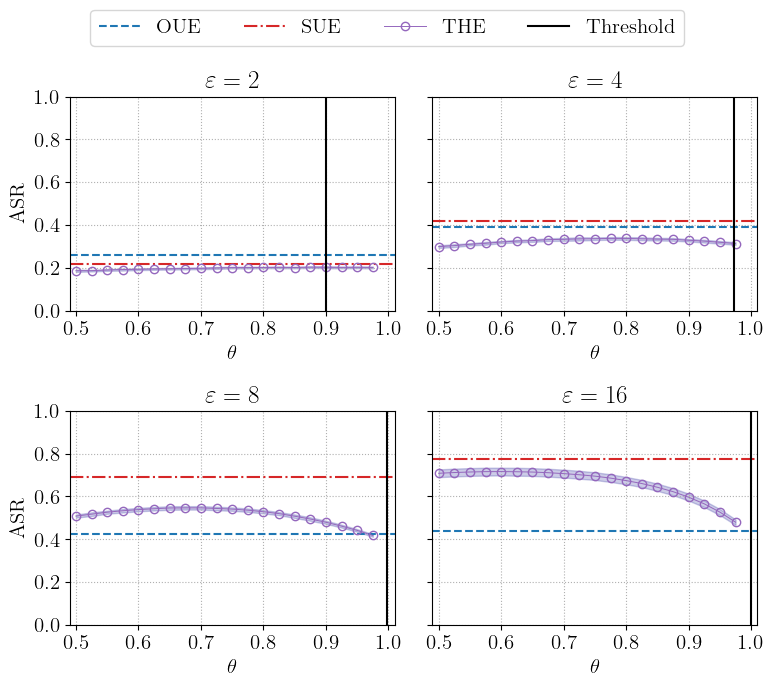}
    \caption{Comparison of ASR of THE, OUE and SUE protocols. The threshold vertical line corresponds to the one defined in Lemma~\ref{lem:theta_ineq}.}
    \label{fig:privacy_the_ue}
\end{figure}

We notice in Figure~\ref{fig:privacy_the_ue} that when $\theta$ is above the threshold defined in Lemma~\ref{lem:theta_ineq}, OUE leaks more than THE (a consequence of the refinement $O_{\varepsilon}\sqsubseteq T_{\varepsilon, \theta}$). However, below the threshold, there is no refinement, and this fact can be observed in the graphs of $\varepsilon = 8$ and $\varepsilon = 16$, where THE is leaking more information than OUE.

On the other hand, Lemma~\ref{lem:sue_ref_the} says that, no matter the gain function, SUE always leaks more than THE, and this fact is confirmed in all four graphs of Figure~\ref{fig:privacy_the_ue}.

As a second example, we turn our attention to Figure~\ref{fig:utility_the_ue}. We observe that OUE is always more useful than SUE, and that SUE, in turn, is always more useful than THE, for the values of $\varepsilon$ shown in the graphs. These experimental results corroborate the theoretical refinement relationships among THE, OUE, and SUE, as well as the results from~\cite{tianhao2017}.
Saying that a protocol $A$ refines a protocol $B$ means that $B$ leaks at least as much information as $A$. In terms of utility, this corresponds to $B$ being consistently more useful than $A$, but at the same time, less private.

\begin{figure}[h]
    \centering
    \includegraphics[width=1\linewidth]{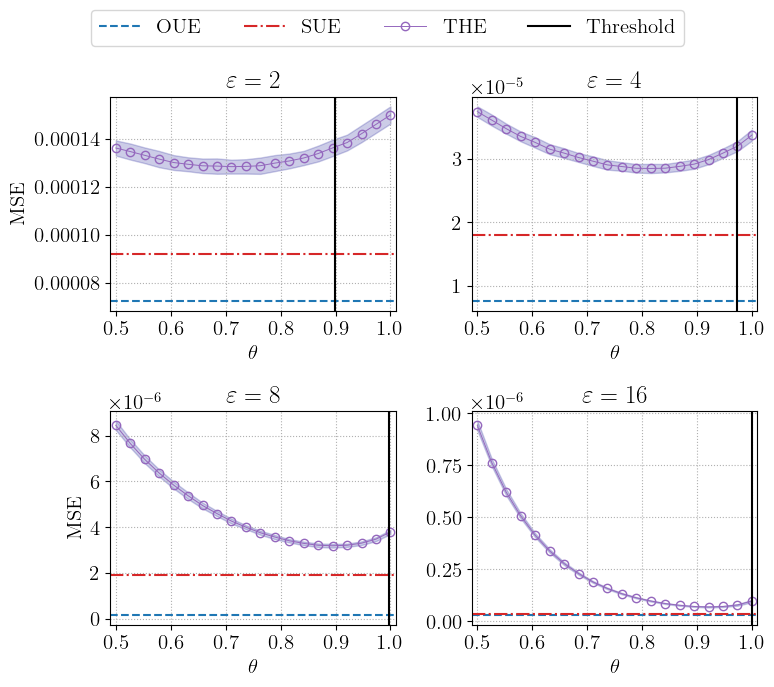}
    \caption{Comparison of the Mean Squared Error (MSE) of an adversary estimating the frequency of values in the data (i.e., the histogram) under the THE, OUE, and SUE protocols. The vertical threshold line corresponds to the value defined in Lemma~\ref{lem:theta_ineq}, and the shaded area around the empirical is the variance.}
    \label{fig:utility_the_ue}
\end{figure}

\subsection{Comparing Bayes Capacity of all LDP protocols}
We begin the analysis of Bayes capacities by examining the OUE and SUE curves in Figure~\ref{fig:capacity_all_protocols}. As shown in Lemma~\ref{lem:ouesue}, neither protocol refines the other, a fact reflected by the intersection of their curves at approximately $\varepsilon = 7$. In contrast, SUE refines THE for all values of $\varepsilon$, and therefore its Bayes capacity is always higher than that of THE.

The binary version of Local Hashing (BLH) maps a domain $[k]$ to a binary domain. Figure~\ref{fig:capacity_all_protocols} shows that the Bayes capacity of BLH is very small for all computed values of $\varepsilon$, indicating that the protocol provides strong privacy. However, this comes at the cost of utility, since most of the information is lost when the domain is reduced to a binary one.

In our experimental evaluation, we observed convergences of Bayes capacities for the LDP protocols analyzed in this work. In summary, the Bayes capacities of GRR, SS, SUE, and THE converge to $k$ as $\varepsilon$ increases; BLH converges to $2$ (because of $g=2$); and OLH and OUE converge to approximately $k/2$.



\begin{figure}[h]
    \centering
    \includegraphics[width=1\linewidth]{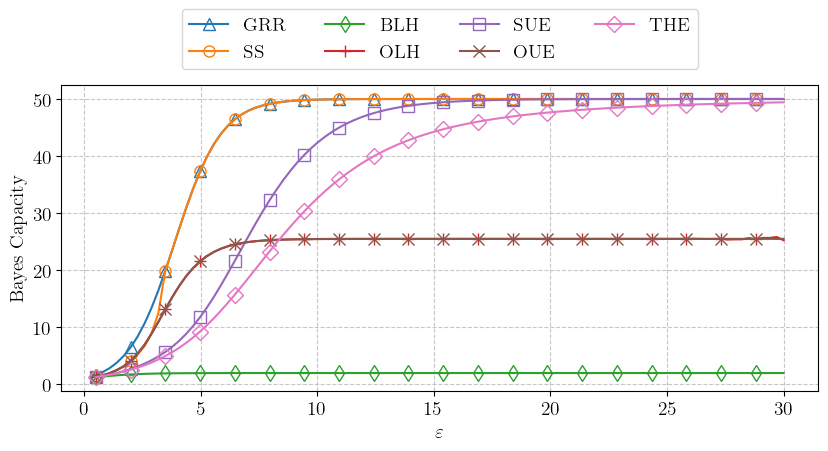}
    \caption{Comparison of the Bayes capacity of all protocols for $k\,{=}\,50$. For THE, $\theta\,{=}\,3/4$, for BLH $g\,{=}\,2$, and all other protocol-specific parameters are set to their optimal values.}
    \label{fig:capacity_all_protocols}
\end{figure}

\subsection{Analysis of Local Hashing}
We present now a comparison between Eqn~\eqref{eq:exp_asr_lh_wrong}, provided in \cite{Gursoy2022}, and our new formulation in Eqn~\eqref{eq:exp_asr_lh}.

\begin{figure}[h]
    \centering
    \includegraphics[width=1\linewidth]{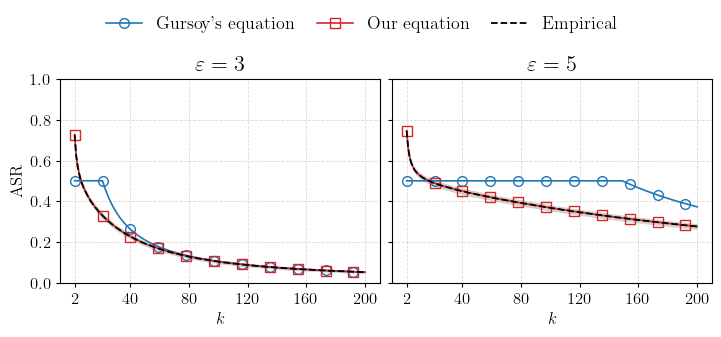}
    \caption{Comparison between Gursoy's~\cite{Gursoy2022} Eqn~\eqref{eq:exp_asr_lh_wrong} and our Eqn~\eqref{eq:exp_asr_lh}. The graph presents the ASR for $\varepsilon=3$ and $\varepsilon=5$ over different values of $k$. The experimental evaluation was repeated over 1000 iterations, simulating a dataset containing 1000 users.}
    \label{fig:comparison_gursoy}
\end{figure}

Figure~\ref{fig:comparison_gursoy} shows that for large values of $k$, the experimental evaluation closely matches both equations. For small $k$, however, Eqn~\eqref{eq:exp_asr_lh_wrong} can deviate substantially from the experimental results. The larger the $\varepsilon$, the larger is the $k$ for a good approximation of Eqn~\eqref{eq:exp_asr_lh_wrong}. In contrast, Eqn~\eqref{eq:exp_asr_lh} matches the experimental evaluation for all values of $\varepsilon$.

\section{Conclusion} \label{sec:conclusion}

In this paper, we have presented a principled and robust framework for the analysis and comparison of Local Differential Privacy (LDP) protocols through the lens of Quantitative Information Flow (QIF). By modeling LDP mechanisms as information-theoretic channels, we move beyond the traditional reliance on the single parameter $\varepsilon$ and utility-specific metrics, which often fail to capture the nuances of privacy protection against diverse adversarial models.

Our main contributions include the formal bridging of the LDP and QIF literatures, demonstrating that common inference attacks, such as data reconstruction, directly coincide with standard QIF notions of Bayes vulnerability. Furthermore, by leveraging the theory of channel refinement and Blackwell ordering, we have established a partial order that allows for metric-independent comparisons of protocols. This refinement-based approach has allowed us to uncover significant discrepancies in current LDP classifications; specifically, we have identified instances where protocols previously characterized as ``optimal'' in terms of estimation error are, in fact, strictly dominated by, or incomparable to, alternative mechanisms when viewed through a rigorous privacy-leakage lens.

Additionally, our channel-based treatment of seven state-of-the-art protocols enabled us to identify and correct analytical errors in existing literature~\cite{Gursoy2022} regarding the expected success rate of reconstruction attacks on Local Hashing (LH) mechanisms. The resulting QIF-based treatment provides the community with a more accurate understanding of the privacy-utility trade-offs inherent in these widely deployed systems.

\section{Perspectives for Future Research}
\label{sec:future_directions}

The intersection of QIF and LDP offers a fertile ground for several future research directions:

\textbf{Complex Data Collection Tasks:} While this paper focuses on frequency estimation, there is a critical need to extend QIF analysis to more complex tasks such as \textit{heavy hitter discovery}~\cite{Bassily2015} and \textit{multidimensional data analysis}~\cite{Arcolezi2022,Filho2023}. 
These tasks involve intricate secret spaces where the attacker's objectives can be much more sophisticated.
    
\textbf{Compositional Analysis:} A vital open question is the study of the \textit{compositional properties} of refinement relations. Investigating whether the dominance of one protocol over another is preserved under sequential or adaptive composition is essential for real-world deployments where users contribute data across multiple rounds.
    
\textbf{Expansion to Metric-LDP:} For applications involving location or numerical data, extending this framework to \textit{metric-based LDP}~\cite{ChatzikokolakisABP13} (e.g., Geo-indistinguishability \cite{AndresBCP13}) would be highly beneficial. This requires a formal treatment of ``closeness" between secrets within the QIF leakage metrics.

\textbf{Automated Formal Verification:} To mitigate the risk of analytical errors in manual proofs, it would be valuable to develop  \textit{automated verification tools}, such as probabilistic model checking or symbolic execution, to provide formal certification of refinement relations.

\section*{Author Contributions}
\noindent 
Natasha Fernandes and Ramon G. Gonze are the lead authors of this work and carried out the majority of the technical development. Nataliia Bielova and Héber H. Arcolezi initiated the project and developed the initial analyses, including the posterior-vulnerability formulation and the channel-matrix representations of LDP protocols. Natasha Fernandes formalized the analysis, establishing connections to $f$-DP, refinement relations, and Bayes capacity. Ramon G. Gonze developed the channel-matrix analysis for LH protocols, derived the corrected Bayes-capacity-based leakage analysis for LH protocols, and contributed to the implementation and experiments. Catuscia Palamidessi provided theoretical supervision and validated the formal results. All authors contributed to the writing and review of the paper.

\section*{Acknowledgments}
\noindent 
The work of Héber H. Arcolezi has been partially supported by the French National Research Agency (ANR), under contracts: ``ANR-24-CE23-6239'' and ``ANR-23-IACL-0006''. The work of Catuscia Palamidessi has been partially supported by the project ELSA of the HORIZON EUROPE Framework Programme   (project number 101070617). The collaboration between Natasha Fernandes and  Catuscia Palamidessi has been supported by the project IDEAL of the Inria ``Equipes Associées'' program.

%% file: A_appendix.tex

\subsection{Leakage Analysis using Bayes Capacity}
After experimentally observing a difference of ASR given by Equations~\eqref{eq:exp_asr_lh_wrong} and \eqref{eq:exp_asr_lh}, we investigated the derivation of Eq.~\eqref{eq:exp_asr_lh_wrong} in \cite{Gursoy2022}, and we found a mistake on it. The fraction 

\begin{align}
    \frac{\bbe\left[|\calu_{H_{\ell},x'_{\ell}}|\right] - 1}{\bbe\left[|\calu_{H_{\ell},x'_{\ell}}|\right]}\label{eq1}
\end{align}

\noindent in Eq.~(24) of the paper corresponds to: Given the true value and reported values are equal, and given the adversary will guess a random value from the subset $\calu_{H_{\ell},x'_{\ell}}$, this fraction is intended to be the adversary's expected chance of guessing \emph{wrongly} the true value. However, the correct way to write this expected chance is 

\begin{align}
    \bbe\left[
    1 - \frac{1}{|\calu_{H_{\ell},x'_{\ell}}|}
    \right],\label{eq2}
\end{align}

\noindent and we can see that Equations~\eqref{eq1} and \eqref{eq2} are not equivalent. 

\subsection{Additional Figures}

\begin{figure}[h]
    \centering \includegraphics[width=0.8\linewidth]{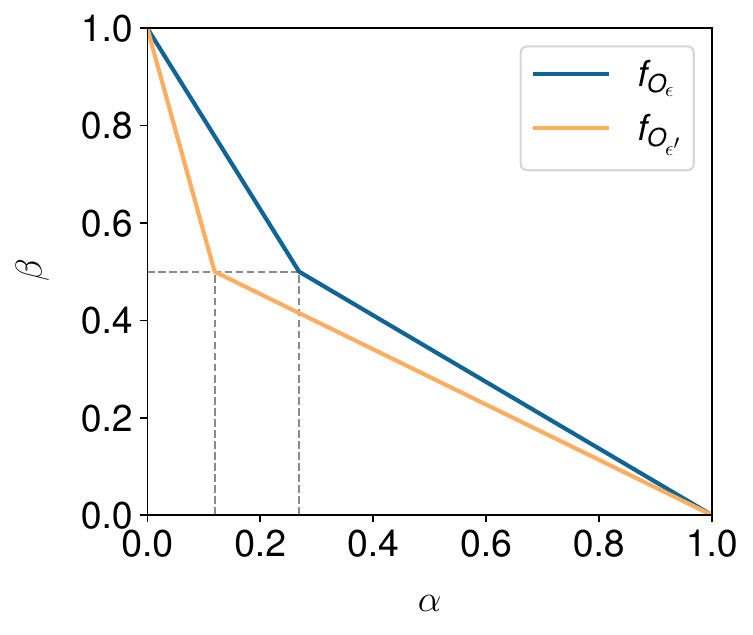}
    \caption{Example showing ordering between $f$-privacy trade-off functions for OUE mechanisms when $\alpha$ values are ordered, i.e. $\alpha > \alpha'$ but $\beta = \beta'$.}\label{fig:OUE_ref}
\end{figure}

\subsection{Proofs omitted from the main body of the paper}

\lemmaequiv*

    \begin{proof}
It is already known that (1) implies (2) (see \cite{alvim2020science}). We will prove that (2) implies (1), (2) implies (3) and (3) implies (2).

\noindent\textbf{(2) implies (1):} 

For this, we need to prove a final requirement for refinement, that there exists an averaging that takes the posteriors of $C$ onto the posteriors of $D$, averaged using the outer probabilities.

We write $u$ for the uniform prior. Note that we can write the posteriors of $u \triangleright C$ as:
{\renewcommand*{\arraystretch}{1.8}
\[
  \delta_{C1} =  \begin{spmatrix}{p_C}
     \frac{1-\alpha_C}{1 - \alpha_C + \beta_C} \\
     \frac{\beta_C}{1 - \alpha_C + \beta_C} 
   \end{spmatrix} 
      ~~~~~   
   \delta_{C2} =  \begin{spmatrix}{q_C}
     \frac{\alpha_C}{1 + \alpha_C - \beta_C} \\
     \frac{1 - \beta_C}{1 + \alpha_C - \beta_C} \\
   \end{spmatrix}
\]}
Similarly for $D$. Now, we are given that the posteriors are generated from a uniform prior, and therefore they must average to $u$. That is:
\begin{align}
	p_C{\cdot} \delta_{C1} + q_C{\cdot} \delta_{C2} = u \label{eqn:C_posteriors} \\	
	p_D{\cdot} \delta_{D1} + q_D{\cdot} \delta_{D2} = u  \label{eqn:D_posteriors} 
\end{align}

We are also given that the posteriors $\delta_{D1}, \delta_{D2}$ lie inside the convex hull formed by $\delta_{C1}, \delta_{C2}$. Therefore there must exist non-negative values $\lambda_1, \lambda_2, \gamma_1, \gamma_2$ such that:
\begin{align}
   \lambda_1 {\cdot} \delta_{C1} + \lambda_2 {\cdot} \delta_{C2} = \delta_{D1} \label{eqn:lambdas} \\
   \gamma_1 {\cdot} \delta_{C1} + \gamma_2 {\cdot} \delta_{C2} = \delta_{D2} \label{eqn:gammas} \\
   \lambda_1 + \lambda_2 = 1 \label{eqn:lambda_sum} \\
   \gamma_1 + \gamma_2 = 1 \label{eqn:gamma_sum}
\end{align}

Multiplying eqns (\ref{eqn:lambdas}), (\ref{eqn:gammas}) by $p_D, q_D$ gives:
\begin{align*}
    (p_D \lambda_1) {\cdot} \delta_{C1} + (p_D \lambda_2) {\cdot} \delta_{C2} = p_D \delta_{D1} \\
   (q_D \gamma_1) {\cdot} \delta_{C1} + (q_D \gamma_2) {\cdot} \delta_{C2} = q_D \delta_{D2} \\
\end{align*}

Summing gives:
\begin{align*}
    (p_D \lambda_1) {\cdot} \delta_{C1} + (p_D \lambda_2) {\cdot} \delta_{C2}~~ + \\
      (q_D \gamma_1) {\cdot} \delta_{C1} + (q_D \gamma_2) {\cdot} \delta_{C2}  &= p_D{\cdot}  \delta_{D1} + q_D {\cdot} \delta_{D2} \\
         &= p_C{\cdot}  \delta_{C1} + q_C{\cdot}  \delta_{C2} \\ 
         \qquad\qquad \mathrm{(from (\ref{eqn:C_posteriors}), (\ref{eqn:D_posteriors}))} 
\end{align*}

Since $\delta_{C1}, \delta_{C2}$ are linearly independent (as vectors), we must have that 
\begin{align*}
	p_D \lambda_1 + q_D \gamma_1 = p_C \\
	p_D \lambda_2 + q_D \gamma_2 = q_C 
\end{align*}

And from eqns (\ref{eqn:lambda_sum}), \ref{eqn:gamma_sum}) we have that
\begin{align*}
    p_D \lambda_1 + q_D \lambda_2 = p_D \\
    q_D \gamma_1 + q_D \gamma_2 = q_D
\end{align*}

Therefore we have shown an averaging of vectors $\delta_{C1}, \delta_{C2}$ onto vectors $\delta_{D1}, \delta_{D2}$, which is the additional requirement for refinement to hold.

\noindent \textbf{(2) implies (3):}

As above, we write the posteriors of $u \triangleright C$ as:
{\renewcommand*{\arraystretch}{1.8}
\[
  \delta_{C1} =  \begin{spmatrix}{p_C}
     \frac{1-\alpha_C}{1 - \alpha_C + \beta_C} \\
     \frac{\beta_C}{1 - \alpha_C + \beta_C} 
   \end{spmatrix} 
      ~~~~~   
   \delta_{C2} =  \begin{spmatrix}{q_C}
     \frac{\alpha_C}{1 + \alpha_C - \beta_C} \\
     \frac{1 - \beta_C}{1 + \alpha_C - \beta_C} \\
   \end{spmatrix}
\]}
and similarly for $D$. Since the posteriors of $u \triangleright D$ lie inside the convex hull of the posteriors of $u \triangleright C$, and these posteriors lie on a straight line (the probability simplex on 2 secrets), then we must have:
\begin{align*}
	&&   \frac{1 - \alpha_C}{1 - \alpha_C + \beta_C}    &\geq \frac{1 - \alpha_D}{1 - \alpha_D + \beta_D}  \\
	\implies && \frac{1 - \alpha_C + \beta_C}{1 - \alpha_C} &\leq \frac{1 - \alpha_D + \beta_D}{1 - \alpha_D} \\
	\implies && 1 + \frac{\beta_C}{1 - \alpha_C} &\leq 1+ \frac{\beta_D}{1 - \alpha_D} \\
	\implies && \frac{\beta_C}{1 - \alpha_C} &\leq \frac{\beta_D}{1 - \alpha_D}
\end{align*}
By a similar argument it is easy to show that $\frac{1 - \beta_C}{\alpha_C} \geq \frac{1 - \beta_D}{\alpha_D}$.

\noindent \textbf{(3) implies (2):}

This follows exactly the reverse argument of the above.
\end{proof}

\sufficiency*

\begin{proof}
Let $\alpha_C \leq \alpha_D$ and $\beta_C \leq \beta_D$. Then
\begin{align*}
  \frac{\beta_C}{1 - \alpha_C} &\leq  \frac{\beta_D}{1 - \alpha_C}  ~~&\mathrm{(Given ~ \beta_C \leq \beta_D) } \\
  & \leq \frac{\beta_D}{1 - \alpha_D} ~~&\mathrm{(since~ \frac{1}{1-\alpha_C} \leq \frac{1}{1 - \alpha_D})}
\end{align*}

Similarly we have that:
\begin{align*}
  \frac{1 - \beta_C}{\alpha_C} &\geq  \frac{1 - \beta_C}{\alpha_D}  ~~&\mathrm{(Given ~\alpha_C \leq \alpha_D)} \\
  & \geq \frac{1 -\beta_D}{\alpha_D} ~~&\mathrm{(since ~ 1-\beta_C \geq 1 - \beta_D)}
\end{align*}

The result follows from Lemma \ref{lem:equiv} (3).
\end{proof}

\kronrefinement*

\begin{proof}
    That $A \refby B$ implies $A^{\otimes k} \refby B^{\otimes k}$ was already shown in \cite{AlvimCCS23}.
    Now, $A \refby B$ means that there exists a witness $W$ such that $A{\cdot}W = B$. That is, $\sum_y A_{x, y} W_{y, z} = B_{x, z}$ for every row $x$ and column $z$ in $B$. But, if we remove row $x'$ from $A$, then we still have $\sum_y A_{x, y} W_{y, z} = B_{x, z}$ for every row $x$ in $B$ except for row $x'$ which has been removed from $A$. Therefore, if we remove the corresponding row $x'$ from $B$ then the equation $\sum_y A_{x, y} W_{y, z} = B_{x, z}$ still holds for all $x$ in $A$. In other words, we have refinement. And so we can conclude that $A^{\otimes k} \refby B^{\otimes k}$ implies $A^{\otimes k}_{\mathrm{hot}} \refby B^{\otimes k}_{\mathrm{hot}}$.   
\end{proof}

\neyman*

\begin{proof}
    Given a 2x2 channel of the form 
    \[
   C =  \begin{bmatrix}
        1 - \alpha & \alpha \\
        \beta & 1 - \beta 
    \end{bmatrix}
    \]
    with $1 - \alpha > \beta$,
    the Neyman-Pearson lemma gives that the most powerful test at significance level $\alpha$ has power 1 - $\beta$. Notice that this is true for $T_{\varepsilon, \theta}$ by construction. That is, we have $\alpha = q = \frac{1}{2} e^{\frac{-\varepsilon \theta}{2}}$ and $\beta = 1-p = \frac{1}{2}e^{\frac{\varepsilon (\theta - 1)}{2}}$. \\
    Next, if $\varepsilon \leq \varepsilon'$ then we can see that 
    \[
    	\frac{1}{2} e^{\frac{-\varepsilon \theta}{2}} \geq \frac{1}{2} e^{\frac{-\varepsilon' \theta}{2}}
\]
 and that 
 \[
 	\frac{1}{2}e^{\frac{\varepsilon (\theta - 1)}{2}} = \frac{1}{2}e^{\frac{-\varepsilon (1 - \theta )}{2}} \geq \frac{1}{2}e^{\frac{- \varepsilon' (1 - \theta )}{2}}~.
\] And so by Lemma \ref{lem:sufficiency} we have that $T_{\varepsilon', \theta} \refby T_{\varepsilon, \theta}$.
\end{proof}

\ouesue*

\begin{proof}
    We use the $f$-privacy trade-off functions to compare $O_{\epsilon}$ and $S_\epsilon$. This time we compare them using Lemma \ref{lem:equiv} since the condition for Lemma \ref{lem:sufficiency} does not hold. 
  For $S_\epsilon$ we have $\alpha_S = \beta_S = \frac{1}{e^{\epsilon/2} + 1}$ and for $O_\epsilon$ we have $\alpha_O = \frac{1}{e^\epsilon + 1}, \beta_O = 1/2$. Therefore we have:
\begin{align*}
  \frac{\beta_O}{1 - \alpha_O} -  \frac{\beta_S}{1 - \alpha_S} &= \frac{1/2}{e^\epsilon / (e^\epsilon + 1)} - \frac{1}{e^{\epsilon/2}} \\
          &= \frac{e^\epsilon + 1}{2 e ^\epsilon} - \frac{1}{e^{\epsilon/2}} \\
          &= \frac{e^\epsilon + 1 - 2 e^{\epsilon/2}}{2 e^\epsilon} \\
          &= \frac{(e^{\epsilon/2} - 1)^2}{2 e^\epsilon} \\
          &\geq 0
\end{align*}
Thus $\frac{\beta_O}{1 - \alpha_O} \geq \frac{\beta_S}{1 - \alpha_S}$.
Next we have that:
\begin{align*}
  \frac{1 - \beta_S}{\alpha_S} - \frac{1 - \beta_O}{\alpha_O} &= e^{\epsilon/2} - \frac{1/2}{1 / (e^\epsilon + 1} \\
     &= e^{\epsilon/2} - \frac{e^\epsilon + 1}{2} \\
     &= \frac{2 e^{\epsilon/2} - e^\epsilon - 1}{2} \\
     &= \frac{- (e^{\epsilon/2} - 1)^2}{2} \\
     &\leq 0
\end{align*}
And so $\frac{1 - \beta_S}{\alpha_S} \leq \frac{1 - \beta_O}{\alpha_O}$.
Thus the condition in Lemma \ref{lem:equiv} (3) does not hold, and so refinement does not hold between $S_\epsilon$ and $O_\epsilon$, except in the trivial case when equality holds, which means $\epsilon = 0$, corresponding to the channel that leaks nothing.
\end{proof}

\ouethe*

\begin{proof}
   We make use of Lemma \ref{lem:equiv} (3). We first recall the bitwise mechanisms for $T_{\varepsilon, \theta}$ and $O_{\varepsilon}$:

\[
   T_{\varepsilon, \theta} ~=~ \begin{bmatrix}
       1-\frac{1}{2} e^{\frac{-\varepsilon \theta}{2}} & \frac{1}{2} e^{\frac{-\varepsilon \theta}{2}} \\
       \frac{1}{2}e^{\frac{\varepsilon (\theta - 1)}{2}} & 1 - \frac{1}{2}e^{\frac{\varepsilon (\theta - 1)}{2}}
   \end{bmatrix} 
   ~~~
   O_{\varepsilon} ~=~
       \begin{bmatrix}
        \frac{e^{\varepsilon}}{e^{\varepsilon}+1} & \frac{1}{e^{\varepsilon}+1} \\
        1/2 & 1/2
    \end{bmatrix}
\]
We note that these are in the right format to apply Defn~\ref{def:alpha_beta} (by construction). We will label the tradeoff points $(\alpha_T, \beta_T)$ and $(\alpha_O, \beta_O)$ for $T_{\varepsilon, \theta}$ and $O_{\varepsilon}$ respectively. Now, considering the second inequality, we have
\begin{equation}\label{oue_1}
    \frac{1 - \beta_O}{\alpha_O} ~=~ \frac{\frac{1}{2}}{1/(e^\varepsilon + 1)} ~=~ \frac{e^\varepsilon + 1}{2}
\end{equation}
\begin{equation}\label{the_1}
    \frac{1 - \beta_T}{\alpha_T} ~=~ \frac{1 - \frac{1}{2} e^{-\varepsilon(1-\theta)/2}}{\frac{1}{2} e^{-\varepsilon \theta / 2}} ~=~ e^{\varepsilon \theta/2}(2 - e^{-\varepsilon(1 - \theta)/2})
\end{equation}
And so we deduce:
\begin{align*}
    \eqref{oue_1} - \eqref{the_1} &=~ \frac{e^\varepsilon + 1}{2} - e^{\varepsilon \theta/2}(2 - e^{-\varepsilon(1 - \theta)/2}) \\
    &=~ e^{\varepsilon\theta} e^{-\varepsilon/2} - 2e^{\varepsilon\theta/2} + \frac{e^\varepsilon + 1}{2} 
\end{align*}
Now, let $x = e^{\varepsilon\theta/2}$. And so we continue,
\begin{align*}
    \qquad &=~ e^{-\varepsilon/2}x^2 -2x + \frac{e^\varepsilon+1}{2} \\
    &=~ e^{-\varepsilon/2} \Bigl( x^2 - 2e^{\varepsilon/2}x + e^{\varepsilon/2}(\frac{e^\varepsilon + 1}{2}) \Bigr) \\
    &=~ e^{-\varepsilon/2} \Bigl( (x-e^{\varepsilon/2})^2 + e^{\varepsilon/2}(\frac{e^\varepsilon + 1}{2}) - e^\varepsilon \Bigr) 
\end{align*}
The quadratic term is always positive, so we consider the term on the RHS. Letting $y = e^{\varepsilon/2}$ we have that
\begin{align*}
    e^{\varepsilon/2}(\frac{e^\varepsilon + 1}{2}) - e^\varepsilon ~&=~ y(\frac{y^2 + 1}{2}) - y^2 \\
    &=~ \frac{1}{2} (y^3 - 2y^2 + y) \\
    &=~ \frac{1}{2} y (y-1)^2 \\
    &\geq 0 \textrm{ for } y \geq 0
\end{align*}
And so we have shown that $\eqref{oue_1} - \eqref{the_1} \geq 0$, thus we have the second condition of Lemma \ref{lem:equiv} (3) that holds. We need now to check the first inequality. We begin:
\begin{equation}\label{oue_next}
    \frac{\beta_O}{1 - \alpha_O} ~=~ \frac{e^\varepsilon + 1}{2e^\varepsilon} 
\end{equation}
\begin{align}
  \frac{\beta_T}{1 - \alpha_T} ~&=~ \frac{\frac{1}{2} e^{-\varepsilon(1 - \theta)/2}}{1 - \frac{1}{2} e^{-\varepsilon \theta/2}} \nonumber \\
  &=~ \frac{e^{-\varepsilon(1 - \theta)/2}}{2 - e^{-\varepsilon\theta/2}} \nonumber \\
  &=~ \frac{e^{\varepsilon\theta}e^{-\varepsilon/2}}{2e^{\epsilon\theta/2} - 1} \nonumber \\
  &=~ \frac{x^2 e^{-\varepsilon/2}}{2x - 1} 
\end{align}
where in the last line we set $x = e^{\varepsilon\theta/2}$. Then,
\begin{align*}
    \frac{\beta_T}{1 - \alpha_T} - \frac{\beta_O}{1 - \alpha_O} ~&=~ \frac{x^2 e^{-\varepsilon/2}}{2x - 1} - \frac{e^\varepsilon + 1}{2e^\varepsilon} \\
    &=~ \frac{2e^{\varepsilon/2}x^2 - x(2e^\varepsilon + 2) + (e^\varepsilon + 1)}{2e^\varepsilon (2x - 1)} 
\end{align*}
This is $\geq 0$ when the numerator is $\geq 0$ (since the denominator is always positive). The numerator is quadratic in $x$, so solving gives
\begin{align*}
    x \geq \frac{(2e^\varepsilon + 2) + \sqrt{(2e^\varepsilon + 2)^2 - 8e^{\varepsilon/2}(e^\varepsilon + 1)}}{4 e^{\varepsilon/2}} \textrm{ or } \\
    x \leq \frac{(2e^\varepsilon + 2) - \sqrt{(2e^\varepsilon + 2)^2 - 8e^{\varepsilon/2}(e^\varepsilon + 1)}}{4 e^{\varepsilon/2}}
\end{align*}
We simplify the expression inside the square root:
\begin{align*}
    &(2e^\varepsilon + 2)^2 - 8e^{\varepsilon/2}(e^\varepsilon + 1) \\
    &~~=~ 4 + 4e^{2\varepsilon} + 4e^{\varepsilon} - 8e^{3\varepsilon/2} - 8e^{\varepsilon/2} \\
    &~~=~ 4(1 + e^{2\varepsilon} + e^{\varepsilon} - 2e^{3\varepsilon/2} - 2e^{\varepsilon/2} \\
    &~~=~ 4\Bigl( (e^{\varepsilon} + 1)^2 - 2e^{\varepsilon/2}(e^\varepsilon + 1) \Bigr) \\
    &~~=~ 4(e^{\varepsilon} + 1)(e^{\varepsilon/2} - 1)^2
\end{align*}
And so, substituting $x = e^{\varepsilon\theta/2}$ we have that
\begin{align*}
    & e^{\varepsilon\theta/2} \geq \frac{e^{\varepsilon} + 1 + (e^{\varepsilon/2} - 1)\sqrt{e^\varepsilon+1}}{2e^{\varepsilon/2}} \\
    \implies & \theta \geq \frac{2 \ln{(e^{\varepsilon} + 1 + (e^{\varepsilon/2} - 1)\sqrt{e^\varepsilon+1})} - \varepsilon - 2\ln2}{\varepsilon}
\end{align*}
and
\begin{align*}
    & e^{\varepsilon\theta/2} \leq \frac{e^{\varepsilon} + 1 - (e^{\varepsilon/2} - 1)\sqrt{e^\varepsilon+1}}{2e^{\varepsilon/2}} \\
    \implies & \theta \leq \frac{2 \ln{(e^{\varepsilon} + 1 - (e^{\varepsilon/2} - 1)\sqrt{e^\varepsilon+1})} - \varepsilon - 2\ln2}{\varepsilon}
\end{align*}
\end{proof}

\thesue*

\begin{proof}
We make use of Lemma \ref{lem:equiv} (3). We have:
\begin{align*}
  \frac{\beta_T}{1 - \alpha_T} - \frac{\beta_S}{1 - \alpha_S} &= \frac{1/2 e^{-\epsilon(1 - \theta)/2}}{1 - 1/2 e^{-\epsilon\theta/2}} - \frac{1}{e^{\epsilon/2}} \\
   &= \frac{e^{-\epsilon(1 - \theta)/2}}{2 - e^{-\epsilon \theta/2}} - \frac{1}{e^{\epsilon/2}} \\
   &= \frac{e^{\epsilon/2}(e^{-\epsilon(1-\theta)/2}) - (2 - e^{-\epsilon\theta/2})}{e^{\epsilon/2}(2 - e^{-\epsilon\theta/2})} \\
   &= \frac{e^{\epsilon \theta / 2} - 2 + e^{-\epsilon \theta / 2}}{e^{\epsilon/2}(2 - e^{-\epsilon\theta/2})} \\
   &= \frac{e^{\epsilon\theta} - 2e^{\epsilon\theta/2} + 1}{e^{\epsilon(1 - \theta)/2}(2 - e^{-\epsilon\theta/2})} \\
   &= \frac{(e^{\epsilon\theta/2} - 1)^2}{e^{\epsilon(1 - \theta)/2}(2 - e^{-\epsilon\theta/2})} \\
   &\geq 0
\end{align*}
And finally,
\begin{align*}
  \frac{1 - \beta_S}{\alpha_S} - \frac{1 - \beta_T}{\alpha_T} &= e^{\epsilon/2} - \frac{1 - 1/2 e^{-\epsilon(1 - \theta)/2}}{1/2 e^{-\epsilon\theta/2}} \\
  &= e^{\epsilon/2} - e^{\epsilon\theta/2}(2 - e^{-\epsilon(1 - \theta)/2}) \\
  &= e^{\epsilon/2} - 2 e^{\epsilon\theta/2} + e^{-\epsilon/2}e^{\epsilon\theta} \\
  &= \frac{e^{\epsilon\theta} - 2e^{\epsilon\theta/2}e^{\epsilon/2} + e^\epsilon}{e^{\epsilon/2}} \\
  &= \frac{(e^{\epsilon\theta/2} - e^{\epsilon/2})^2}{e^{\epsilon/2}} \\
  &\geq 0
\end{align*}
And so by Lemma \ref{lem:equiv} we have that $S_{\epsilon} \sqsubseteq T_{\epsilon, \theta}$.
\end{proof}

\asrlh*
\begin{proof}



\noindent We are going to derive the  posterior vulnerability of the $\lhm$ channel that corresponds exactly to the ASR. The channel $\lhm^{\varepsilon} = EG$ can be decomposed in the cascade of channels $E$ and $G$:
\begin{enumerate}
    \item First encode $x$ (channel $E$).
    \item Apply GRR to the encoded value (channel $G$).
\end{enumerate}

\noindent\textbf{Definition of $E$}: The channel is defined as $E\,{:}\,\mathcal{X}\,{\rightarrow}\,\mathbb{D}\mathcal{Y}$, where $\mathcal{Y}\,{=}\,\{(y_h,y_e)~|~y_h\in\mathscr{H} \text{ and } y_e\in[g]\}$. The channel receives a secret $x\in\mathcal{X} = [k]$ as input and outputs a pair $y = (y_h,y_e)$, where $y_h = \{(1,w_i), \ldots, (k,w_k)~|~w_k \in [g]\}$ is a hash function chosen uniformly at random from $\mathscr{H}$ and $y_e$ is the encoded value of $x$. The channel matrix is then defined as

\begin{align}
    {E}_{x,y} = \pr[(y_h,y_e)\,|\,x] = \begin{cases}
        \frac{1}{g^{k}} & \text{ , if $y_h(x) = y_e$}\\
        0 & \text{ , if $y_h(x) \neq y_e$}
    \end{cases}
\end{align}

\noindent\textbf{Definition of $G$}: The channel is defined as $G : \mathcal{Y} \rightarrow \mathbb{D}\mathcal{Z}$, where $\mathcal{Z} = \mathcal{Y}$. The channel receives a pair $y = (y_h,y_e)$ as input and outputs a pair $z = (z_h,z_p)$, where $z_h = y_h$ is the hash function received as input and $z_p = GRR(y_e)$ is the output of GRR for the encoded value of $x$ (aka perturbed version of $y_e$). The definition is as follows

\begin{align}
    G^{\varepsilon}_{y,z} = \begin{cases}
        \frac{e^\varepsilon}{e^\varepsilon + g - 1} & \text{ , if $y_h = z_h$ and $y_e = y_z$}\\
        \frac{1}{e^\varepsilon + g - 1} & \text{ , if $y_h = z_h$ and $y_e \neq y_z$}\\
        0 & \text{ , if $y_h \neq z_h$}.
    \end{cases}
\end{align}

\noindent\textbf{Definition of $\lhm^{\varepsilon}$}: The channel is defined as a cascading $\lhm^{\varepsilon} = EG$, where 
\begin{align*}
    \lhm^{\varepsilon}_{x,z} &= \sum\limits_{y} E_{x,y} \cdot G^{\varepsilon}_{y,z}\\
    &= \begin{cases}
        \frac{1}{g^k}\cdot\frac{e^\varepsilon}{e^\varepsilon+g-1} & \text{ , if $z_h(x) = z_p$}\\
        \frac{1}{g^k}\cdot\frac{1}{e^\varepsilon+g-1} & \text{ , if $z_h(x) \neq z_p$}.\\
    \end{cases}\\
\end{align*}

\noindent In short, $\lhm^{\varepsilon}$ is a channel that, given an input $x\in\mathcal{X}$, maps it to a pair $z = (z_h,z_p)$ where $z_h$ is a hash function chosen randomly and $z_p = GRR(z_h(x))$.

\noindent \textbf{Posterior vulnerability:}
Let $g$ be the Bayes gain function such that the set of actions $\mathcal{W}\,{=}\,\mathcal{X}$ and $g(w,x)\,{=}\,1 \text{ if } w\,{=}\,x \text{ or } 0 \text{ otherwise}$. We have that

\begin{align}
    V_g[\pi \rhd \lhm^{\varepsilon}] &= \sum\limits_{z} \max\limits_{w} \sum\limits_{x}\pi_x \cdot \lhm^{\varepsilon}_{x,z}\cdot g(w,x)\nonumber\\
    \intertext{Given that $g(w,x) \neq 0 \iff w = x$, we can eliminate the inner sum:}
    V_g[\pi \rhd \lhm^{\varepsilon}] &= \sum\limits_{z} \max\limits_{x} \left(\frac{1}{k} \cdot \lhm^{\varepsilon}_{x,z}\right) \nonumber\\
    \intertext{We can expand the sum on $z$ by each element of the pair $(z_h,z_p)$:}
    &= \frac{1}{k} \sum\limits_{z_p=1}^{g} \sum\limits_{z_h\,\in\,\mathscr{H}}  \max\limits_{x} \lhm^{\varepsilon}_{x,(z_h,z_p)} \nonumber\\
    \intertext{Once there is a fixed $z_p$, there are two possible values for $\max\limits_{x}\lhm^{\varepsilon}_{x,(z_h,z_p)}$: $\frac{p}{g^k}$ and $\frac{q}{g^k}$. The last one happens when there is no $x$ such that $z_h(x)\,{=}\,z_p$, so $\left|\{(1,e_1), \ldots, (k,e_k) | e_i \in [g]\backslash y_s\}\right| = (g-1)^k$. The first one happens when there is at least one $x$ such that $z_h(x) = z_p$, thus $\left|\{(1,e_1), \ldots, (k,e_k) | e_i \in [g] \text{ and } \exists j : e_j = z_s\}\right| = g^k\,{-}\,(g{-}1)^k$. We can split the second summation $\sum_{z_h}$ among these two case as following:}\nonumber\\
    V_g[\pi \rhd \lhm^{\varepsilon}] &= \frac{1}{k} \sum\limits_{z_p=1}^{g} \left( \frac{p}{g^k}(g^k-(g-1)^k) + \frac{q}{g^k}(g-1)^k \right)\nonumber\\
    &= \frac{p(g^k - (g-1)^k) + q(g-1)^k}{kg^k} \sum\limits_{z_p=1}^{g} 1 \nonumber\\
    &= \frac{(pg^k - p(g-1)^k + q(g-1)^k)g}{kg^k} \nonumber\\
    &= \frac{pg^k + (g-1)^k(q-p)}{kg^{k-1}} \nonumber\\
    &= \frac{e^\varepsilon g^k + (g-1)^k (1-e^\varepsilon)}{(e^\varepsilon + g -1)(kg^{k-1})}. \nonumber
\end{align}

\end{proof}